\documentclass[letterpaper, 10 pt, conference]{ieeeconf}  
\usepackage{amsmath}
\usepackage{amssymb}  
\usepackage{bm}
\usepackage{mathtools}
\usepackage{xcolor}
\usepackage{graphicx}
\usepackage{booktabs}
\usepackage{cite}
\usepackage{url}
\usepackage{hyperref}
\usepackage{todonotes}
\let\labelindent\relax
\usepackage{enumitem}
\usepackage{multirow}
\usepackage{algorithm}
\usepackage{algpseudocode}

\usepackage{fontawesome5}
\usepackage[most]{tcolorbox}
\newtcbox{\linkpill}{
  on line,
  box align=base,
  colback=gray!10,
  colframe=gray!10,
  arc=1.5mm,
  boxrule=0pt,
  left=2pt,
  right=2pt,
  top=1pt,
  bottom=1pt
}
\definecolor{slsorange}{HTML}{D97706}

\newcommand{\R}{\mathbb{R}}
\newcommand{\nx}{n_x}
\newcommand{\nuu}{n_u}
\newcommand{\Phix}{\mathbf{\Phi}^{\mathrm{x}}}
\newcommand{\Phiu}{\mathbf{\Phi}^{\mathrm{u}}}
\newcommand{\Phit}{\mathbf{\Phi}}
\newcommand{\Psix}{\mathbf{\Psi}^{\mathrm{x}}}
\newcommand{\Psiu}{\mathbf{\Psi}^{\mathrm{u}}}
\newcommand{\Psit}{\mathbf{\Psi}}
\newcommand{\HPhix}{\mathbf{\hat\Phi}^{\mathrm{x}}}
\newcommand{\HPhiu}{\mathbf{\hat\Phi}^{\mathrm{u}}}
\newcommand{\HPhit}{\mathbf{\hat\Phi}}
\newcommand{\TPhix}{\mathbf{\tilde\Phi}^{\mathrm{x}}}
\newcommand{\TPhiu}{\mathbf{\tilde\Phi}^{\mathrm{u}}}
\newcommand{\TPhit}{\mathbf{\tilde\Phi}}

\newcommand{\RR}{\mathbb{R}}
\newcommand{\cX}{\mathcal{A}}
\newcommand{\f}{p}

\newcommand{\augx}{\bm{\alpha}}
\newcommand{\nomx}{\tilde{\bm{\alpha}}}
\newcommand{\pathx}{{\augx}_{\nomx}}

\newcommand{\J}{J}
\newcommand{\Hess}{H}

\newcommand{\rem}{\bm{r}}
\newcommand{\Rem}{\mathcal{R}}
\newcommand{\remlow}{\underline{\rem}}
\newcommand{\remup}{\overline{\rem}}

\newcommand{\rlow}{\underline{r}}
\newcommand{\rup}{\overline{r}}

\newcommand{\lowerW}{\mathbf{\underline{W}}}
\newcommand{\upperW}{\mathbf{\overline{W}}}
\newcommand{\lowerb}{\mathbf{\underline{b}}}
\newcommand{\upperb}{\mathbf{\overline{b}}}

\newtheorem{theorem}{Theorem}
\newtheorem{lemma}{Lemma}

\newtheorem{proposition}{Proposition}
\newtheorem{corollary}{Corollary}

\newcommand{\delete}[1]{}

\IEEEoverridecommandlockouts                              

\usepackage{graphicx} 

\title{\LARGE \bf
GPU-Parallel Linearization Error Bounds for Real-Time \\ Robust Optimal Control of Nonlinear and Neural Network Dynamics
}

\author{Jeffrey Fang$^\star$ \and Keyi Shen$^\star$ \and Anutam Srinivasan \and Glen Chou
\thanks{$^\star$ Equal contribution. All authors are with the Georgia Institute of Technology, Atlanta, GA, USA, 30308.
        {\tt\small \{jfang301, kshen84, asrinivasan350, chou\}@gatech.edu}}%
}

\begin{document}

\maketitle
\thispagestyle{empty}
\pagestyle{empty}

\begin{abstract}

\looseness-1This paper studies real-time robust optimal control for uncertain nonlinear systems, where linear time-varying (LTV) approximations make planning tractable but require sound linearization error bounds (LEBs) to guarantee robust constraint satisfaction. We develop tight, differentiable, GPU-parallel LEBs for LTV approximations of nonlinear and neural network (NN) dynamics. For analytic dynamics, we introduce path-based Hessian bounds that are tighter than standard interval methods. For NN dynamics, we derive certified LEBs using NN verifier-generated affine relaxations and local Jacobian corrections. We adapt a GPU-parallel system-level synthesis LTV-based robust control solver to be compatible with these LEBs by extending it to handle right-invertible disturbance matrices and non-zero-centered disturbance sets for tight zonotopic uncertainty propagation. Our method, GPUSLS-LEO, enables online optimization of robust feedback policies that account for linearization error, producing tight, formally verified reachable tubes. On complex nonlinear and NN dynamics up to 168 state dimensions, our method can compute robust control policies on the GPU at rates up to 67 Hz, reducing solve times and conservativeness relative to baselines while preserving formal guarantees and real-time performance. 

\centering{\linkpill{\href{https://trustworthyrobotics.github.io/gpusls-leo/}{\textcolor{slsorange}{\faGlobe\ Project Website}}}
\quad
\linkpill{\href{https://github.com/trustworthyrobotics/gpusls-leo}{\textcolor{slsorange}{\faGithub\ Code (GitHub)}}}}

\end{abstract}

\section{Introduction}

\looseness-1Safe real-time nonlinear control under uncertainty is essential for the resilient operation of robots, drones, and spacecraft. This motivates methods for solving the \textit{robust nonlinear optimal control problem} (RNOCP), which seeks a control policy that \textit{robustly} guarantees constraint satisfaction under worst-case disturbances. In practice, RNOCP solvers typically optimize (1) a nominal trajectory and (2) a stabilizing feedback controller to track it under disturbance. Robust constraint satisfaction is enforced by tightening the nominal constraints using an overapproximation of the closed-loop tracking-error \textit{reachable tube}, ensuring that the closed-loop dynamics remain in a safe subset of the state and input space.
However, exact reachability via error propagation through nonlinear dynamics \cite{majumdar2017funnel, althoff2013reachability, bansal2017hamilton} is generally intractable in real time. To address this, \cite{leeman2025robust, schafer2025robust, leister2025robust, althoff2008reachability} approximate the nonlinear dynamics along a nominal trajectory with a linear time-varying (LTV) system and compute reachable sets using the LTV dynamics. System-level synthesis (SLS) \cite{DBLP:journals/arc/AndersonDLM19, leeman2025robust} is one such LTV robust control framework that can be solved in real time over long horizons for high-dimensional systems \cite{leeman2024fast, fang2026safe}.

Despite strong performance, LTV-based methods only guarantee robust constraint satisfaction given a bound on the \textit{linearization error}, which captures the deviation between the true nonlinear dynamics and their LTV approximation. Existing methods often use global bounds \cite{leeman2025robust, kim2024joint}, which yield conservative (loose) reachable tubes. Tighter \textit{local} bounds around the nominal trajectory are often possible, but must be computed \textit{efficiently} for real-time planning and informative enough to \textit{guide} the optimizer toward safe trajectories with minimal linearization error. Thus, the challenge is to develop bounds that are simultaneously (1) tight, (2) real-time computable, and (3) differentiable, so an RNOCP solver can use them to reshape the nominal trajectory and tighten the reachable tubes. As existing methods do not achieve all three, many LTV-based RNOCP solvers must omit linearization error, forgoing formal guarantees for performance \cite{zhan2026robustly, messerer2021efficient}.

\begin{figure}[t]
    \centering
    \includegraphics[width=\linewidth]{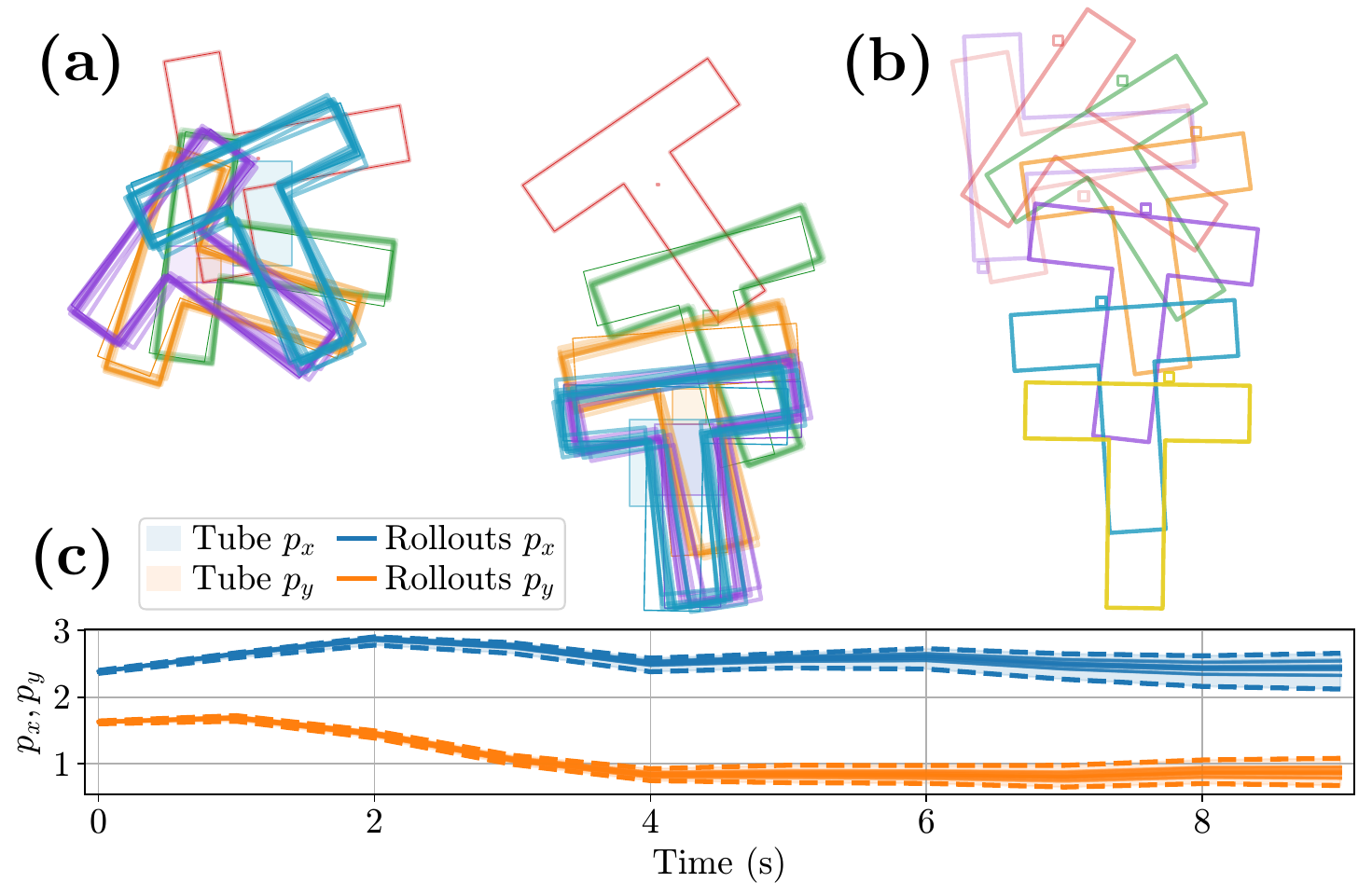}
    \caption{\textbf{(a)}: Robust tubes from GPUSLS-LEO for a neural T pusher system and disturbed rollouts for a rotation and a push trajectory against random disturbances. \textbf{(b)}: MPC rollout using our method on real dynamics, successfully moving the T to the goal. \textbf{(c)}: Robust tubes of $x$ and $y$ position showing tight tubes and all simulated rollouts staying within the tubes. 
    } \label{fig:neuralt_rollouts}
\end{figure}

To close this gap, we propose a family of tight, differentiable linearization error bounds (LEBs) for RNOCP solvers with nonlinear dynamics described by analytic functions and neural networks (NNs), with a GPU-parallel implementation in JAX. To make these LEBs compatible with the real-time SLS solver of \cite{fang2026safe, leeman2024fast}, we extend the solver 1) to accommodate right-invertible disturbance matrices, enabling direct zonotopic propagation of the LEBs while avoiding the conservative overapproximations in \cite{leeman2024fast}; and 2) to handle disturbance sets that are not zero-centered. Finally, we show empirically that these bounds enable real-time, GPU-based computation of formally verified RNOCP solutions, with limited conservativeness for systems with up to 168 states and solve times as low as 15 ms. Our contributions are:
\begin{itemize}[leftmargin=*, itemindent=0pt]
    \item A set of tight, provably-sound linearization error bounding methods, including path-based interval bounds for analytic systems and NN verifier-based bounds for NN dynamics.
    \item A novel nonlinear SLS formulation that includes linearization error through tight zonotopic propagation, including treatment of right-invertible disturbance matrices and non-zero-centered disturbance sets, guaranteeing robust constraint satisfaction for the original nonlinear system.
    \item \looseness-1A GPU-accelerated nonlinear SLS solver, \textbf{GPUSLS-LEO}, that uses a formulation of these linearization bounds in JAX to enable differentiable, GPU-parallel evaluation and gradient-based robust nonlinear optimal control.
    \item Evaluation on analytic and NN dynamics models, maintaining real-time control rates (up to 67 Hz) for systems up to 168D and problems with $\approx 2\times 10^5$ decision variables.
\end{itemize}

\section{Related Work}

\looseness-1While linearization-free reachability methods exist \cite{majumdar2017funnel, althoff2013reachability, bansal2017hamilton}, direct error propagation through nonlinear dynamics is generally intractable in real time, especially beyond $\approx$10 states \cite{bansal2017hamilton, majumdar2017funnel}. Thus, RNOCP solvers often rely on linearized error propagation for efficiency \cite{manchester2017dirtrel, messerer2021efficient}, but introduce linearization error that is difficult to tightly bound, and is often ignored \cite{manchester2017dirtrel, diehl2006approximation, messerer2021efficient, fang2026safe}, underapproximated via sampling \cite{kim2024joint, leeman2025robust}, or handled by assuming access to a linearization bound oracle \cite{houska2012robust, villanueva2017robust, cannon2009successive, cannon2011robust, richards2005robust}. 
RNOCP solvers that explicitly compute LEBs typically rely on global Lipschitz estimates \cite{kim2024joint} or interval arithmetic \cite{althoff2008reachability, schafer2025robust, leister2025robust, rungger2018accurate}. Interval arithmetic has been used both to compute reachable sets \cite{limon2005robust, shen2026parallel} and to support reachability by constructing linear differential inclusion bounds (LDI) for nonlinear dynamics \cite{harapanahalli2024immrax}. In particular, \cite{harapanahalli2025linear} uses Jacobian overbounding to obtain tighter LDIs for reachability, but is not used for control design or directly compute LEBs. Computationally, existing linearization bounding methods are offline \cite{leeman2025robust}, slow, or non-differentiable \cite{althoff2008reachability}, limiting their use for planning. To close these gaps, our method uses interval arithmetic and JAX to compute tight differentiable LEBs for direct use in optimal control. 
Moreover, we also adapt neural-network (NN) verification tools \cite{zhang2018efficient} to derive affine bounds on linearization error for NN dynamics. 

\looseness-1As LEBs can be slow to compute \cite{leeman2025robust}, they are often generated once offline \cite{leeman2025robust, kim2024joint} but are loose if we wish to bound error around \textit{one trajectory}, as in RMPC. Thus, practitioners often omit linearization error, yielding heuristic reachable tubes that trade formal guarantees for performance \cite{zhan2026robustly}.
To address this, we show that tight, local, GPU-parallel LEBs can be computed online while providing both real-time performance on high-dimensional systems and hard guarantees.

\section{Preliminaries and Problem Statement}
\looseness-1\noindent\textbf{Notation:} We denote $I_n \in \mathbb{R}^{n\times n}$ as the identity matrix, the Frobenius norm of a matrix $A \in \mathbb{R}^{m \times n}$ as $\| A \|_{\mathcal{F}} := \sqrt{\text{Trace}(A^\top A)}$. Let $[N]:=\{0,...,N-1\}$, and $[M,N]:=\{M,...,N\}$ for $M,N \in \mathbb{N}$. For a matrix $A \in \mathbb{R}^{n \times p}$ we define the row-wise $\ell_1$-norm $\| A \| _{1,\text{r}} \in \mathbb{R}^{n}$ by $\|A\|_{1, \text{r}} \coloneqq \left[ \|A_{1,:} \|_1, \ldots, \|A_{n,:} \|_1\right]^\top$. For $a \in \mathbb{R}^n$, we define $\text{diag}(a) \in\mathbb{R}^{n\times n}$ as its diagonal matrix. Given $A \in \mathbb{R}^{m \times n}$ with linearly independent rows, we define $A^{\ddagger} \in \mathbb{R}^{n \times m}$ such that $AA^{\ddagger} = I_{m}$, where $A^{\ddagger}:=A^{\top}(AA^{\top})^{-1}$. We denote the Minkowski sum of two sets $A$ and $B$ as $A \oplus B := \{a + b \mid a \in A, b \in B\}$. For a function $p(x)$ with $x \in \mathbb{R}^n$ (or explicitly $p\left((x_1,\dots,x_n)\right)$), we define {\small$\partial_{j}p(\alpha):=\frac{\partial p}{\partial x_j} |_{x = \alpha}$} and {\small$\partial_{jk}p(\alpha):=\frac{\partial^2 p}{\partial x_j\partial x_k} |_{x = \alpha}$} as the first and second partial derivatives with respect to the $j$-th and $k$-th components of the vector input. We define $\mathbb{I}_n$ as the set of all $n$-dimensional intervals. Let, $\Delta,\hat\Delta \in \mathbb{I}_n$, then $\Delta\cdot\hat\Delta := \{\delta \cdot  \hat\delta\mid \delta\in\Delta,\hat\delta\in\hat\Delta\}$. 

\noindent\textbf{Definitions:} We consider uncertain nonlinear dynamics
\begin{equation}
    \label{eq:nmpc_dynamics}
    x_{k+1} = f(x_k, u_k) + E(x_k) w_k,
\end{equation} 
where $x_k \in \mathcal{X}\subseteq \R^{\nx}$ is the system state at time step $k$, $u_k \in \mathcal{U}\subseteq \R^{\nuu}$ denotes the control input, $f: \mathbb{R}^{n_x} \times \mathbb{R}^{n_u} \rightarrow \mathbb{R}^{n_x}$ the dynamics function,
$E: \mathbb{R}^{n_x} \rightarrow \mathbb{R}^{n_x \times n_x}$ the disturbance scaling function, and $w_k \in \mathcal{E}_{n_x} \coloneqq \{ w \in \mathbb{R}^{n_x}, \|w\|_\infty \leq 1 \}$ the disturbance, normalized to be contained in a unit $\ell_{\infty}$-ball.

We aim to design a controller $\pi(\cdot)$ that solves the following robust nonlinear optimal control problem (RNOCP):
\begin{subequations} \label{eq:robust_nocp}
    \begin{align}
        \min_{\pi(\cdot)} \quad &J\left(\bar{x}, \pi(\cdot)\right) \\
        \text{s.t.} \quad & x_{k+1} = f(x_k, u_k) + E(x_k)w_k, \quad \forall k \in [N],
        \\
                          & x_0 = \bar{x}_0, \\ 
                          & u_k = \pi_k(x_{0},\dots,x_{k}), \quad \forall k \in [N], \\
                          & g_k(x_k, u_k) \leq 0,  \quad \forall w_k \in \mathcal{E}_{n_x},\quad\forall k \in [N], \label{eq:robust_nocp_con_running}\\
                          & g^f(x_N) \leq 0, \quad \forall w_k \in \mathcal{E}_{n_x},\quad\forall k \in [N].\label{eq:robust_nocp_con_terminal}
    \end{align}
\end{subequations}
\looseness-1where $\pi:=\{\pi_i\}_{i=0}^{N-1}$ is a sequence of causal control policies, $\bar x_0 \in \mathbb{R}^{\nx}$ is the initial state, functions $g_k: \mathbb{R}^{n_x} \times \mathbb{R}^{n_u} \rightarrow \mathbb{R}^{n_c}$ denote stagewise state-input constraints, and $g^f: \mathbb{R}^{n_x} \rightarrow \mathbb{R}^{n_f}$ is the terminal constraint. We define the closed-loop reachable set $\Omega_k :=(\Omega_k^x, \Omega_k^u)$ of states and controls as 
\begin{equation}\label{eq:reach}\small
    \hspace{-7pt}\Omega_k :=
\left\{
(x_k, u_k) \;\middle|\;
\begin{aligned}
&x_{j+1} = f(x_j,u_j) + E(x_j) w_j,\ \ \forall j\in[k],\\
&\qquad \forall w_j \in \mathcal{E}_{\nx},\quad \forall j \in [k], \\
& u_j=\pi_j(x_0, \ldots, x_j), \quad \forall j\in [k+1].
\end{aligned}
\right\}.
\end{equation}

\looseness-1
Because \eqref{eq:robust_nocp} is an intractable infinite-dimensional problem, RNOCP solvers often simplify it by computing (A) a nominal state-input trajectory and (B) a tracking controller around that trajectory. In this paper, we use one such method, GPUSLS \cite{fang2026safe}, which adapts \cite{leeman2024fast} for real-time solution on the GPU. We then extend \cite{fang2026safe} to account for linearization error in Sec. \ref{sec:method_sls}. Our approach uses SLS, which we now review.

\subsection{System Level Synthesis (SLS)}\label{sec:sls}

Consider the following uncertain LTV dynamics
\begin{equation}\label{eq:dynamics_ltv}
    x_{k+1} = A_kx_k + B_ku_k+E_kw_k.
\end{equation}
SLS is a control design framework that optimizes over causal \textit{disturbance feedback} controllers
\begin{equation}\label{eq:disturbance_feedback}
u_k = v_k + \sum_{j=0}^{k-1} \Phiu_{k,j} E_jw_j,
\end{equation}
\looseness-1with nominal control $v_k \in \R^{\nuu}$, i.e., we assign a disturbance feedback matrix $\Phiu_{k,j}\in \R^{\nuu\times \nx}$ for each disturbance $E_jw_j$ and control $u_k$ for $k>j$.
For LTV dynamics \eqref{eq:dynamics_ltv}, it can be shown using algebraic manipulations \cite{DBLP:journals/arc/AndersonDLM19} that the resulting closed-loop state sequence can be expressed as
\begin{equation}\label{eq:SLS_state}
x_k= z_k   + \sum_{j=0}^{k-1} \Phix_{k,j} E_jw_j,~ z_0 = \bar x_0,
\end{equation}
\looseness-1where $z_k \in \mathbb{R}^{\nx}$ is the nominal state and the closed-loop response \mbox{$\Phix_{k,j}\in \R^{\nx\times \nx}$} captures the effect of disturbance $w_j$ on state $x_k$.
Starting with $\Phix_{j+1,j} = I_{n_x}$, SLS propagates the disturbance via $\Phix_{k+1,j} = A_k \Phix_{k,j} + B_k\Phiu_{k,j}$, 
for all $j \in [N]$ and $k \in [j+1,N-1]$. Since \eqref{eq:disturbance_feedback}--\eqref{eq:SLS_state} are the \textit{true} closed-loop trajectory under a given disturbance sequence, the \textit{exact} closed-loop reachable sets are $\Omega_x^k = z_k + \textstyle\bigoplus_{j=0}^{k-1}\Phix_{k,j}E_j \mathcal{E}_{\nx}$ and $\Omega_k^u = v_k+\textstyle\bigoplus_{j=0}^{k-1}\Phiu_{k,j}E_j \mathcal{E}_{\nx}$. 
To apply the methods of Sec.~\ref{sec:sls} to the nonlinear dynamics \eqref{eq:nmpc_dynamics}, we approximate \eqref{eq:nmpc_dynamics} with the nominal dynamics $z_{k+1} = f(z_k,v_k)$ and LTV error dynamics along a nominal state-control trajectory $\mathbf{z} := \{z_k\}_{k=0}^N \in \mathbb{R}^{\nx \times (N+1)}$, $\mathbf{v} := \{v_k\}_{k=0}^{N-1} \in \mathbb{R}^{\nuu \times N}$. For a nominal $(z_k, v_k) \in \mathbb{R}^{\nx} \times \mathbb{R}^{\nuu}$, we construct an LTV approximation of the form \eqref{eq:dynamics_ltv} by defining
\begin{equation}\label{eq:linearizations}
     A_k := \left.\frac{\partial f}{\partial x}\right|_{(z_k,v_k)}, B_k := \left.\frac{\partial f}{\partial u}\right|_{(z_k,v_k)}, E_k := E(z_k).
\end{equation}

Denote $\delta z_k := x_k - z_k$ and $\delta v_k := u_k - v_k$. Then, the residual error $d_k(x_k, u_k, w_k)$ between \eqref{eq:nmpc_dynamics} and its first-order Taylor approximation about $(z_k, v_k)$, i.e. $f(z_k, v_k) + A_k \delta z_k + B_k \delta v_k$, can be written as,
\begin{subequations}\label{eq:linearization_error}
\begin{align}
d_k(x_k,&u_k,w_k) := E_k w_k + \underbrace{r_k^{\mathrm{lin}} + r_k^{\mathrm{dist}}}_{\qquad:=r(x_k,u_k,w_k)}, \label{eq:linearization_error_all}\\[-5mm]
r_k^{\mathrm{dist}} &:= \bigl(E(x_k)-E_k\bigr) w_k,\label{eq:linearization_error_dist}\\
r_k^{\mathrm{lin}} &:=
f(x_k,u_k)
- f(z_k,v_k)
- A_k \delta z_k 
- B_k \delta v_k. \label{eq:linearization_error_lagrange}
\end{align}
\end{subequations}
Specifically, \eqref{eq:linearization_error_lagrange} captures the Taylor remainder of the nonlinear dynamics $f$, while \eqref{eq:linearization_error_dist} accounts for the error induced by the state dependence of the disturbance matrix $E$. Given this setup, we study the following problems.

\subsection{Problem Statement}
\looseness-1\noindent\textbf{Problem 1 (Linearization Error Bound).}
For all $k \in [N]$, find a set-valued overapproximation $\bar{\mathcal{E}}_k \subseteq \mathbb{R}^{\nx}$ and associated scaling matrix $\bar R_k \in \mathbb{R}^{n_x \times n_x}$ and offset $\bar c_k \in \mathbb{R}^{\nx}$ where
\begin{equation}\label{eq:lin_bound}
\begin{aligned}
r_k(x_k,u_k,w_k) &\in \bar{\mathcal{E}}_k \subseteq \big( \bar R_k\mathcal{E}_{\nx} \oplus \bar c_k \big), \\
&\forall (x_k, u_k) \overset{\eqref{eq:reach}}{\in} \Omega_k,\ \forall w_k \in \mathcal{E}_{\nx}.
\end{aligned}
\end{equation}

\noindent\textbf{Problem 2 (Real-time RNOCP).}
Using $\bar{\mathcal{E}}_k$ to account for the discrepancy between \eqref{eq:nmpc_dynamics} and its LTV approximation, efficiently solve the RNOCP \eqref{eq:robust_nocp} using the LTV surrogate model \eqref{eq:dynamics_ltv} to robustly satisfy the constraints \eqref{eq:robust_nocp_con_running} and \eqref{eq:robust_nocp_con_terminal}. 

\section{GPU-Parallel Linearization Error Bounds}\label{sec:method}

To streamline exposition we consider a vector-valued function $\f(\augx): \RR^n \to \RR^m$, over the augmented state $\augx := [x^\top, u^\top]^\top$. Let $\cX \subset \RR^n$ be a compact set and $\nomx \in \cX$ be a nominal point. The linearization error, i.e., the Lagrange remainder of the first-order Taylor expansion of $\f$ at $\nomx$ is, 
\begin{equation}
\rem(\augx;\nomx) := \f(\augx) - \f(\nomx) - \J(\nomx)(\augx - \nomx), \quad \augx \in \cX.
\label{eq:linear_error_def}
\end{equation}
For the specific case when $\f(\augx) = f(x, u)$, \eqref{eq:linear_error_def} coincides with \eqref{eq:linearization_error_lagrange}. Our goal is to compute a sound componentwise interval enclosure $\Rem \in \mathbb{I}_m$ of $\rem(\cX;\nomx)$, i.e., $\rem(\cX;\nomx) \subseteq \Rem$.

\subsection{Linearization Error Bounds for Analytic Dynamics}\label{sec:method_linearization_analytical}

By Taylor's theorem in Lagrange form, for each component $\ell \in \{1,\dots,m\}$, there exists $\xi \in \cX$ such that
\begin{equation}
\rem_\ell(\augx;\nomx) = \tfrac{1}{2} \delta^\top \Hess_\ell(\xi) \delta,
\quad \delta := \augx - \nomx.
\label{eq:quadratic_remainder}
\end{equation}

\looseness-1where $\Hess_\ell(\xi)$ is the $\ell$-th item of the Hessian $\Hess$ at $\xi$. To evaluate the interval enclosure of $\rem_\ell(\cX;\nomx)$, previous work~\cite{althoff2008reachability, schafer2025robust, leister2025robust, rungger2018accurate}, separately evaluate $\delta^\top, \Hess_\ell(\xi)$ and $\delta$ in interval arithmetic, and then perform interval products as in~\eqref{eq:interval_quadratic_form}, yielding the enclosure $\Rem_{\ell, \mathrm{classic}}$:
\begin{equation}
\rem_\ell(\cX;\nomx) \subseteq \Rem_{\ell, \mathrm{classic}} := \tfrac{1}{2} \Delta^\top [\Hess_\ell](\cX) \Delta,
\label{eq:interval_quadratic_form}
\end{equation}
where $\Delta := \cX \ominus \nomx \subset \RR^n$ denotes the displacement set and $[\Hess_\ell](\cX)$ is an interval enclosure of the Hessian over $\cX$ with entries \([\Hess_{\ell,jk}](\cX)\) (where \(\Hess_{\ell,jk}(\augx):=\partial_{jk}\f_\ell(\augx)\)). In practice, $\nomx$ is usually chosen in the center of $\cX$~\cite{althoff2008reachability}, leading to a symmetric $\Delta$ for smaller enclosure of $\rem_\ell(\cX;\nomx)$.

Retaining the tight quadratic form, ~\eqref{eq:quadratic_remainder} can also be first expanded componentwise, followed by performing interval arithmetic on the expansion. It yields~\eqref{eq:expanded_interval_hessian}, where interval $\Delta_j$ denotes the $j$-th slice of $\Delta$, after the symmetry of the interval Hessian is applied. When $\Delta$ is symmetric, $\Rem_{\ell, \mathrm{expand}} \subseteq \Rem_{\ell, \mathrm{classic}}$ is guaranteed because~\eqref{eq:expanded_interval_hessian} preserves the non-negativity of $\delta_k^2$, for all $\delta_k \in \Delta_k$ and for all $k \in \{1,\ldots,n\}$:
\begin{align}
\rem_\ell(\cX;\nomx) \subseteq\; \Rem_{\ell, \mathrm{expand}} = 
&
\textstyle\sum_{\substack{1 \le j < k \le n}} 
[\Hess_{\ell,kj}](\cX)\,\Delta_j\cdot\Delta_k \nonumber\\
+ &\tfrac{1}{2}
\textstyle\sum_{k=1}^n [\Hess_{\ell,kk}](\cX)\,\Delta_k^2.\label{eq:expanded_interval_hessian}
\end{align}

\looseness-1However, the classical bounds in~\eqref{eq:interval_quadratic_form} and~\eqref{eq:expanded_interval_hessian} evaluate all interval enclosures of the Hessian entries over \textit{the full set} $\cX$, causing conservativeness due to dependency and wrapping effects. To mitigate this, we use a \emph{coordinate-wise path} from $\nomx$ to $\augx$ that enables tighter bounds on second-order terms. 
For any $\augx \in \cX$, define the sequence of intermediate points
\begin{equation}
\pathx^{(k)} := (\augx_1, \dots, \augx_k, \nomx_{k+1}, \dots, \nomx_n),
\quad k\in [n+1],
\end{equation}
with $\pathx^{(0)} = \nomx$ and $\pathx^{(n)} = \augx$. This defines a \emph{coordinate-wise} path~\cite{harapanahalli2024immrax, harapanahalli2025linear} that incrementally moves from $\nomx$ to $\augx$ along each coordinate. Correspondingly, define the prefix sets $\cX^{(k)}(\nomx) := \{ \augx \in \cX \mid \augx_{\ell} = \tilde\augx_{\ell}, \forall \ell \in [k+1,n] \}$.

\begin{theorem}[Path-Based Hessian Bound]
\label{thm:path_hessian}
Let $\f:\RR^n \to \RR^m$ be twice continuously differentiable on an interval $\cX \subset \RR^n$, and let $\nomx \in \cX$. Define the displacement set $\Delta := \cX \ominus \nomx$. For each $\augx \in \cX$, define the coordinate-wise path $\augx^{(k)}$
and the corresponding prefix sets $\cX^{(k)}(\nomx)$.
Then, for each component $\ell \in \{1,\dots,m\}$, the linearization error satisfies
\begin{align}\nonumber
\hspace{-6pt}\rem_\ell(\cX;\nomx) \subseteq \Rem_{\ell, \mathrm{path}}&\hspace{-3pt} :=\hspace{-3pt}
\textstyle\sum_{1 \le j < k \le n}
[\Hess_{\ell,kj}](\cX^{(j)}(\nomx))\,\Delta_j\cdot\Delta_k \\
&+ \tfrac{1}{2}
\textstyle\sum_{k=1}^n [\Hess_{\ell,kk}](\cX^{(k)}(\nomx))\,\Delta_k^2.\label{eq:path_based_hessian}
\end{align}
\end{theorem}

\begin{proof}
Fix $\augx \in \cX$ and define $\delta := \augx - \nomx$. Using telescoping along the coordinate-wise path, we have 
$\f_\ell(\augx) - \f_\ell(\nomx)
= \textstyle\sum_{k=1}^n \big( \f_\ell(\pathx^{(k)}) - \f_\ell(\pathx^{(k-1)}) \big)$.
For each $k$, applying Taylor's theorem in 1D yields
\begin{equation}\label{eq:path_diff}
\f_\ell(\pathx^{(k)}) - \f_\ell(\pathx^{(k-1)})
= \partial_k \f_\ell(\pathx^{(k-1)}) \delta_k
+ \tfrac{1}{2} \partial_{kk} \f_\ell(\xi_k)\, \delta_k^2,
\end{equation}
\looseness-1for some $\xi_k \in [\pathx^{(k-1)},\pathx^{(k)}] \subseteq \cX^{(k)}(\nomx)$. 
Summing \eqref{eq:path_diff} over $k$ and subtracting $\nabla \f_\ell(\nomx)^\top \delta$ yields $\f_\ell(\augx) - \f_\ell(\nomx)-\nabla \f_\ell(\nomx)^\top \delta$ on the LHS which is $\rem_\ell(\augx;\nomx)$ and the RHS is in 
$\textstyle\sum_{k=1}^n \big( \partial_k \f_\ell(\pathx^{(k-1)}) - \partial_k \f_\ell(\nomx) \big)\delta_k + \tfrac{1}{2} \textstyle\sum_{k=1}^n \partial_{kk} \f_\ell(\xi_k)\, \delta_k^2$. 
Expanding the gradient difference along the same path for $k$, we have 
$\partial_k \f_\ell(\pathx^{(k-1)}) - \partial_k \f_\ell(\nomx)
= \sum_{j=1}^{k-1}
\big( \partial_k \f_\ell(\pathx^{(j)}) - \partial_k \f_\ell(\pathx^{(j-1)}) \big)$.
By the Mean Value Theorem, for each $j<k$, there exists $\eta_{j,k} \in [\pathx^{(j-1)},\pathx^{(j)}] \subseteq \cX^{(j)}$ such that 
$\partial_k \f_\ell(\pathx^{(j)}) - \partial_k \f_\ell(\pathx^{(j-1)})
= \partial_{kj} \f_\ell(\eta_{j,k})\, \delta_j$.
Substituting yields
$\rem_\ell(\augx;\nomx)
= \sum_{1 \le j < k \le n} \partial_{kj} \f_\ell(\eta_{j,k})\, \delta_j \delta_k
+ \tfrac{1}{2} \sum_{k=1}^n \partial_{kk} \f_\ell(\xi_k)\, \delta_k^2$. 
Since $\eta_{j,k} \in \cX^{(j)}(\nomx)$, $\xi_k \in \cX^{(k)}(\nomx)$, and $\delta \in \Delta$, taking interval enclosures yields~\eqref{eq:path_based_hessian}.
\end{proof}

Thm.~\ref{thm:path_hessian} implies that our path-based Hessian bound is never looser than the expanded interval Hessian bound, as each Hessian entry is evaluated over a subset of the full domain.

\begin{corollary}[Path-Based Hessian is Never Looser]\label{cor:path_not_looser} When $\Delta$ is symmetric, 
the enclosure in~\eqref{eq:path_based_hessian} is guaranteed to be no looser than~\eqref{eq:expanded_interval_hessian} and~\eqref{eq:interval_quadratic_form}. That is,
\begin{equation}
    \Rem_{\ell, \mathrm{path}} \subseteq \Rem_{\ell, \mathrm{expand}} \subseteq \Rem_{\ell, \mathrm{classic}}.
\end{equation}
    
\end{corollary}

\begin{proof}
For each $j \in \{1,\dots,n\}$, the prefix set satisfies $\cX^{(j)} \subseteq \cX$ by construction. Then, for every Hessian entry,
\begin{equation}\nonumber
[\Hess_{\ell,kj}](\cX^{(j)}) \subseteq [\Hess_{\ell,kj}](\cX),
[\Hess_{\ell,kk}](\cX^{(k)}) \subseteq [\Hess_{\ell,kk}](\cX).
\end{equation}

\noindent Multiplying by $\Delta_j \Delta_k$ and $\Delta_k^2$ preserves inclusion and summing the resulting terms yields {$\Rem_{\ell, \mathrm{path}} \subseteq \Rem_{\ell, \mathrm{expand}}$}.
\end{proof}

Unlike $\Rem_{\ell, \mathrm{classic}}$, $\Rem_{\ell, \mathrm{expand}}$ and $\Rem_{\ell, \mathrm{path}}$ preserve componentwise expansion to produce tight, off-centered intervals (via the quadratic form rather than symmetric bounds). We exploit the added tightness in the RNOCP setting (Sec. \ref{sec:method_sls}).

\subsection{Linearization Error Bounds for Neural Dynamics}\label{sec:method_linearization_neural}

For NN dynamics $f$, computing a tight interval of the Hessian can be difficult.
Instead, we use CROWN~\cite{zhang2018efficient,xu2020automatic} to obtain affine bounds of an NN over a box input set, and convert them into an interval enclosure of  \(\rem(\augx;\nomx)\).

\begin{theorem}[CROWN bounds (Thm. 3.2 \cite{zhang2018efficient})]
\label{theo:crown}\looseness-1
Given an NN, $\f: \RR^n \to \RR^m$, and a bounded input \(\augx\in\cX\) where \(\cX\) is an interval, there exist linear lower/upper bounds on \(\f(\augx)\):
\begin{equation}
\lowerW \augx + \lowerb \leq \f(\augx) \leq \upperW \augx + \upperb,
\label{eq:crown_bounds}
\end{equation}
where \(\lowerW, \upperW\in \RR^{m\times n}\) and \(\lowerb, \upperb\in \RR^{m}\).
\end{theorem}

\looseness-1The bound in~\eqref{eq:crown_bounds} does not directly bound \(\rem(\augx;\nomx)\), since the slopes \(\lowerW,\upperW\) generally do not match the Jacobian \(\J(\nomx)\). We address this via further relaxation in Thm.~\ref{thm:crown_remainder}.

\begin{theorem}[Linearization error bound from CROWN]
\label{thm:crown_remainder}
Let \(\f:\RR^n\to\RR^m\) be an NN, \(\cX\subset\RR^n\) be a box set, \(\nomx\in\cX\), \(\J(\nomx)\) be the Jacobian of \(\f\) at \(\nomx\), and let \((\lowerW,\lowerb,\upperW,\upperb)\) satisfy the CROWN bounds~\eqref{eq:crown_bounds} on \(\cX\). Then, \(\forall\augx\in\cX\):
\begin{equation}
\begin{aligned}
\remlow(\augx) \le \rem(\augx;\nomx) \le \remup(\augx),
\end{aligned}
\end{equation}
\begin{equation}
\begin{aligned}
\remlow(\augx)
&:= (\lowerW-\J(\nomx))\augx
+\J(\nomx)\nomx
-\f(\nomx)
+\lowerb,\\
\remup(\augx)
&:= (\upperW-\J(\nomx))\augx
+\J(\nomx)\nomx
-\f(\nomx)
+\upperb.
\end{aligned}
\end{equation}
Then, an interval enclosure of \(\rem(\cX;\nomx)\subseteq[\remlow,\remup]\) is given by
\begin{equation}
\remlow_i
=
\min_{\augx\in\cX} \; \remlow_i(\augx),
\quad
\remup_i
=
\max_{\augx\in\cX} \; \remup_i(\augx),
\quad \forall i \in [1,m],
\label{eq:crown_remainder_interval}
\end{equation}
\end{theorem}
\begin{proof}
Substituting the CROWN lower bound~\eqref{eq:crown_bounds} into the expression for \(\rem(\augx;\nomx)\)~\eqref{eq:linear_error_def}, we yield
$\rem(\augx;\nomx)
\ge
\lowerW\augx+\lowerb-\f(\nomx)-\J(\nomx)(\augx-\nomx)$, 
which simplifies to 
$\rem(\augx;\nomx)
\ge
(\lowerW-\J(\nomx))\augx
+\J(\nomx)\nomx
-\f(\nomx)
+\lowerb
=
\remlow(\augx).$ 
Similarly, we apply the CROWN upper bound.
Thus, \(\forall\augx\in\cX\): 
$\remlow(\augx)\le \rem(\augx;\nomx)\le \remup(\augx)$. 
For bounds valid over \(\cX\), we take the componentwise minimum of $\remlow(\augx)$ and maximum of $\remup(\augx)$ over \(\cX\), as in~\eqref{eq:crown_remainder_interval}. 
\end{proof}

\looseness-1Thm.~\ref{thm:crown_remainder} reconciles the CROWN slopes and the local linearization slope \(\J(\nomx)\), thereby producing a sound remainder bound around the nominal point.
Since \(\remlow(\augx)\) and \(\remup(\augx)\) are affine in \(\augx\), problem~\eqref{eq:crown_remainder_interval} is inexpensive when \(\cX\) is a box set.

\subsection{Zonotope Bounds of $r_k(x_k,u_k,w_k)$}

To obtain the zonotope bound for $r_k(x_k,u_k,w_k)$~\eqref{eq:lin_bound}, we compute the  enclosure $[\rlow_k^{\mathrm{lin}}, \rup_k^{\mathrm{lin}}]$ of the linearization error $r_k^{\mathrm{lin}}$ as in Sec.~\ref{sec:method_linearization_analytical} or~\ref{sec:method_linearization_neural} and the enclosure $[\rlow_k^{\mathrm{dist}}, \rup_k^{\mathrm{dist}}]$ of the disturbance error $r_k^{\mathrm{dist}}$ via interval arithmetic in~\eqref{eq:interval_r_dist}, where $[E](\Omega_k^x)$ is the interval enclosure of $E$ over $\Omega_k^x$:
\begin{equation}
    r_k^{\mathrm{dist}} = \bigl(E(x_k)-E_k\bigr) w_k \in \bigl([E](\Omega_k^x) \ominus E_k\bigr) \mathcal{E}_{\nx}.
    \label{eq:interval_r_dist}
\end{equation}

Finally, we obtain $r_k(x_k,u_k,w_k) \in [\rlow_k, \rup_k]$, where $\rlow_k = \rlow_k^{\mathrm{lin}} + \rlow_k^{\mathrm{dist}}$ and $\rup_k = \rup_k^{\mathrm{lin}} + \rup_k^{\mathrm{dist}}$. Its zonotope bound is given by $\big( \bar R_k \mathcal{E}_{\nx} \oplus \bar c_k \big)$, where $\bar c_k = \tfrac{1}{2}(\rlow_k + \rup_k)$ and $\bar R_k = \tfrac{1}{2}\mathrm{diag}(\rup_k - \rlow_k)$. We denote this mapping compactly as
\begin{equation}\label{eq:zono_map}
    \hspace{-1pt}\bar R_k \hspace{-1pt}=\hspace{-1pt} \Gamma_r(\Omega_k, z_k, v_k; f, E), \hspace{6pt} \bar c_k\hspace{-1pt}=\hspace{-1pt}\Gamma_c(\Omega_k, z_k, v_k; f, E).
\end{equation}
\subsection{Differentiable, Parallel Bound Implementation}
We implement our LEBs in JAX, with the PBH using \texttt{immrax} for interval arithmetic and \texttt{jax\_verify} providing CROWN-based bounds for NN dynamics. JAX offers differentiability and GPU parallelization, enabling \emph{gradient-based refinement} of the nominal point $\nomx$ and scalable \emph{input domain partitioning} to reduce conservativeness for SLS.

\section{Using Linearization Error Bounds in GPUSLS}\label{sec:method_sls}
We discuss the use of our LEBs in SLS (Sec. \ref{sec:method_nonlinear_sls}) and describe an efficient GPU-based implementation (Sec. \ref{sec:gpu_implementation}).

\subsection{Nonlinear SLS Formulation}\label{sec:method_nonlinear_sls}
\looseness-1Following \cite[App. A]{leeman2024fast}, SLS can be extended to nonlinear systems by planning a nominal trajectory for the disturbance-free dynamics $z_{k+1} = f(z_k,v_k)$ and modeling tracking error dynamics via the LTV system of \eqref{eq:dynamics_ltv}. 
To solve \eqref{eq:robust_nocp} approximately, we modify the nonlinear SLS formulation of \cite[App. A]{leeman2024fast} to be compatible with 1) non-zero-centered disturbance sets and 2) the LEBs given in Sec. \ref{sec:method_linearization_analytical}-\ref{sec:method_linearization_neural}:
\begin{subequations}\label{eq:fastsls_problem}
    \begin{align}
        \hspace{-13pt}\min_{\substack{
        \mathbf{z},\, \mathbf{v}, \mathbf{\Phi}
        }} \ \
        & J(\mathbf{z},\mathbf{v}) + \tilde{H}_0(\mathbf{\Phi}, \bm{\mathcal{D}}) \label{eq:sls_problem_obj} \\[-1mm]
        \text{s.t.}\quad \
        & z_{k+1} = f(z_k, v_k), \quad \forall k \in [N],\quad z_0 = \bar{x}_0,\label{eq:sls_problem_dyn} \\[-0.5mm] 
        & \Phix_{k+1, j} = A_k\Phix_{k, j} + B_k\Phiu_{k,j},\ \ \Phix_{j+1,j} = I_{\nx}, \\[-0.5mm]
        & \nonumber \quad\quad\forall j \in [N], \hspace{5pt} \forall k \in [j+1, N - 1], \\[-0.5mm]
        & \bar E_k := \mathcal{Z}_r(z_k, v_k, \bm{\mu}_k, \bm{\sigma}_k),\quad \forall k\in [N], \label{eq:rmpc_linearization_radius}\\
        & \bar{c}_k := \mathcal{Z}_c(z_k, v_k, \bm{\mu}_k, \bm{\sigma}_k),\quad \forall k\in [N], \label{eq:rmpc_linearization_center}\\
        & \bm{\sigma}_k = h^{\sigma}_{k}(\Phit, \bm{\mathcal{D}}),~ \bm{\mu}_k = h^{\mu}_{k}(\Phit, \bm{\mathcal{D}}), \,\, \forall k \in [N],\label{eq:rmpc_tube_width_center}\\[-0.5mm]
        & g_k(z_k, v_k) + h_{k}(\Phit, \bm{\mathcal{D}}) \leq 0, \qquad \forall k \in [N],\hspace{-5pt}\label{eq:rmpc_stage_tightening}\\[-0.5mm]
        & g^f(z_N) + h^f(\mathbf{\Phi}, \bm{\mathcal{D}}) \leq 0\label{eq:rmpc_term_tightening},
    \end{align}
\end{subequations}
where $A_k, B_k$ are the linearized dynamics \eqref{eq:linearizations} at time $k$ and functions $\mathcal{Z}_r (\cdot), \mathcal{Z}_c (\cdot)$ \eqref{eq:rmpc_linearization_radius}--\eqref{eq:rmpc_linearization_center} are the Minkowski sums of the exogenous disturbance and linearization error as a function of the nominal ($z_k$,$v_k$), reachable tube center offset $\bm{\mu}_k \in \mathbb{R}^{n_x + n_u}$, and widths $\bm{\sigma}_k\in \mathbb{R}^{n_x + n_u}$ respectively~\cite{zhan2026robustly, leeman2025robust}, and $\bm{\mathcal{D}} := \{(\bar{c}_j, \bar E_j)\}_{j=0}^{N-1}$. Formally, for all $k \in [N]$ in \eqref{eq:rmpc_linearization_radius}--\eqref{eq:rmpc_linearization_center}, $\bar c_k \in \mathbb{R}^{\nx}$ and $\bar E_k \in \mathbb{R}^{\nx \times 2\nx}$ are given as
\begin{subequations} \label{eq:zonotope_combination}
    \begin{align} 
    \hspace{-4pt}\mathcal{A}_k &= \left([z_k^\top, v_k^\top]^\top + \bm{\mu}_k\right) \oplus \textrm{diag}(\bm{\sigma}_k)\mathcal{E}_{\nx+\nuu},\label{eq:zontope_combination_interval}\\
    \hspace{-4pt}\bar{E}_k 
    & := [E(z_k)\ \ \bar R_k]=\underbrace{
    \begin{bmatrix}
    E(z_k) \ \ \Gamma_r\!\left(
        \mathcal{A}_k, z_k, v_k; f, E
    \right)
    \end{bmatrix}
    }_{\mathcal{Z}_r(z_k, v_k, \bm{\mu}_k, \bm{\sigma}_k)},\label{eq:zonotope_combination_concat}\\[-5mm]
    \hspace{-4pt}\bar{c}_k & := \underbrace{\Gamma_c\!\left(
        \mathcal{A}_k, z_k, v_k;
        f, E
    \right)}_{\mathcal{Z}_c(z_k, v_k, \bm{\mu}_k, \bm{\sigma}_k)},\label{eq:zonotope_combination_center}
    \end{align}
\end{subequations}
where $\bar{c}_k$ is the offset and $\bar{E}_k$ is the generator matrix of a zonotope for the Minkowski sum of the exogenous disturbance and linearization error: $E(z_k)\mathcal{E}_{n_x} \oplus (\bar{c}_k\oplus \bar R_k\mathcal{E}_{n_x}) \equiv \bar{c}_k \oplus \bar{E}_k\mathcal{E}_{2n_x}$. Using \eqref{eq:zonotope_combination_concat} and \eqref{eq:zonotope_combination_center}, the resulting reachable tube overapproximation is, 
\begin{equation}
\begin{aligned}\label{eq:sls_reach_non_center}
    \bar\Omega_k^x := z_k + \textstyle\sum_{j=0}^{k-1}\Phix_{k,j}\bar{c}_j + \textstyle\bigoplus_{j=0}^{k-1}\Phix_{k,j}\bar{E}_j \mathcal{E}_{2\nx},\\\bar\Omega_k^u := v_k+ \textstyle\sum_{j=0}^{k-1}\Phiu_{k,j}\bar{c}_j + \textstyle\bigoplus_{j=0}^{k-1}\Phiu_{k,j}\bar{E}_j \mathcal{E}_{2\nx}.
\end{aligned}
\end{equation}

Robust constraint satisfaction is enforced through constraint tightenings $h_k(\mathbf{\Phi}, \bm{\mathcal{D}})$ and $h^f(\mathbf{\Phi}, \bm{\mathcal{D}})$. In short, these tightenings capture the propagation of unit-normalized disturbances in \eqref{eq:sls_reach_non_center}, which are then scaled by the disturbance bounds $\bar{E}_j$ and shifted by offset $\bar c_j$. For simplicity, we derive the tightenings assuming linear constraints, i.e., $g_k(z_k, v_k) = G_k\begin{bmatrix}z_k,\ v_k\end{bmatrix}^\top$ and $g^f(z_N) = G^fz_N$, for $G_k \in \mathbb{R}^{n_c \times (n_x+n_u)}, G_{N}\in \mathbb{R}^{n_c \times n_x}$. Valid tightenings can be obtained for nonlinear $g_k$ and $g^f$ by including additional error terms \cite[Eq. 22]{zhan2026robustly}. Formally, we define the tube center \eqref{eq:tight_tube_center} and width \eqref{eq:tight_tube_width}, stacked (terminal) constraint tightenings \eqref{eq:tight_constraint}-\eqref{eq:tight_terminal_constraint}, and surrogate system-level response cost \eqref{eq:surrogate_cost_h0} as, 
\begin{align}
    h_k^{\mu}(\mathbf{\Phi}, \bm{\mathcal{D}}) &=
    \textstyle\sum_{j=0}^{k-1}\Phit_{k,j}
    \bar{c}_j\label{eq:tight_tube_center}\\
    h_k^{\sigma}(\mathbf{\Phi},\bm{\mathcal{D}}) &=
    \textstyle\sum_{j=0}^{k-1}
    \big\Vert \Phit_{k,j}\bar{E}_j\big\Vert_{1,\mathrm{r}}\label{eq:tight_tube_width}\\
    h_k(\mathbf{\Phi}, \bm{\mathcal{D}})
    &=
    \textstyle\sum_{j=0}^{k-1}(G_k\Phit_{k,j}\bar{c}_j + 
    \big\Vert G_k\Phit_{k,j}\bar{E}_j \big\Vert_{1,\mathrm{r}})\label{eq:tight_constraint}\\
    h^f(\mathbf{\Phi}, \bm{\mathcal{D}})
    &= \textstyle\sum_{j=0}^{N-1}
    (G_{N}\mathbf{\Phi}^{\mathrm{x}}_{N,j}\bar{c}_j + 
    \big\Vert
    G_{N}\mathbf{\Phi}^{\mathrm{x}}_{N,j} \bar{E}_j\big\Vert_{1,\mathrm{r}})\hspace{-2pt}\label{eq:tight_terminal_constraint}\\
    \tilde{H}_0(\Phit, \bm{\mathcal{D}})&=\textstyle\sum_{j=0}^{N-1}\!(\textstyle\sum_{k=j}^{N-1}\!(\|\tilde Q^{1/2}\Phix_{k,j}\bar{E}_j\|_{\mathcal F}^2\label{eq:surrogate_cost_h0}\\&\hspace{4mm}+\|\tilde R^{1/2}\Phiu_{k,j}\bar{E}_j\|_{\mathcal F}^2)\nonumber+\|\tilde Q_N^{1/2}\Phix_{N,j}\bar{E}_j\|_{\mathcal F}^2)\,,
\end{align}
\looseness-1where $\small{\Phit_{k,j} := \begin{bmatrix}
    {\Phix_{k,j}}^{\top} & \hspace{-2mm}{\Phiu_{k,j}}^{\top}
\end{bmatrix}^{\top}}$, $\Phit$ collects $\Phit_{k,j}$ for all $k,j\in[N]$, and $\tilde{Q}, \tilde{Q}_N  \in \mathbb{S}_{++}^{n_x}, \tilde{R} \in \mathbb{S}_{++}^{n_u}$. The constraint tightenings and tube widths in \eqref{eq:tight_tube_width}-\eqref{eq:tight_terminal_constraint} are derived for the worst case disturbances using the dual norm property \cite{li2026certified}.
The cost $\tilde{H}_0(\Phit,\bm{\mathcal{D}})$ \eqref{eq:surrogate_cost_h0} penalizes uncertainty (the tube widths ~\eqref{eq:tight_tube_width}). 
Following \cite{leeman2024fast, fang2026safe}, we solve \eqref{eq:fastsls_problem} via an iterative process that alternates between finding a (A) nominal trajectory and a (B) robust controller. To solve (A), we find $(\mathbf{z}, \mathbf{v})$ by solving a constraint-tightened NOCP \eqref{eq:nmpc_problem_tightened},
\begin{subequations} \label{eq:nmpc_problem_tightened}
    \begin{align}
        \hspace{-5pt}\min_{\substack{
        \mathbf{z},\, \mathbf{v}
        }} \quad
        & J(\mathbf{z},\mathbf{v}) \label{eq:nmpc_problem_tightened_cost}\\
        \text{s.t.} \quad 
        & z_{k+1} = f(z_k, v_k), \quad  \forall k \in [N], \qquad z_0 = \bar{x}_0,\label{eq:nmpc_problem_tightened_dyn}\\
        & g_k(z_k, v_k) + h_k(\mathbf{\Phi}, \bm{\mathcal{D}})\leq 0,\quad  \forall k \in [N], \\
        & g^f(z_N) + h^{f}(\mathbf{\Phi}, \bm{\mathcal{D}})\leq 0.
    \end{align}
\end{subequations}
Then, \eqref{eq:nmpc_problem_tightened} mirrors the structure of \eqref{eq:robust_nocp} and can be solved using the GPU-parallel method of \cite[Sec. IV]{fang2026safe}. 
After solving (A) \eqref{eq:nmpc_problem_tightened}, we linearize around the resulting $(\mathbf{z}, \mathbf{v})$ and solve for a robust controller (B) by finding a $\mathbf{\Phi}$ that optimizes \eqref{eq:surrogate_cost_h0}: 
\begin{subequations} \label{eq:fastsls_qp}
    \begin{align}
        \hspace{-3pt}\min_{\Phix, \Phiu} \ 
        &\sum_{j=0}^{N-1}
        \sum_{k=j}^{N-1} \| \mathcal{Q}_{k,j} \Phit_{k,j}\bar{E}_j \|^2_{\mathcal{F}} + \|\mathcal{Q}_{N,j} 
         \Phix_{N,j} \bar{E}_j\|^2_{\mathcal{F}}\label{eq:fastsls_qp_cost} \\
        \text{s.t.}\ \ & \Phix_{k+1,j} = A_k \Phix_{k,j} + B_k \Phiu_{k,j}, \label{eq:fastsls_qp_x_prop}\\
        & \Phix_{j+1,j} = I_{n_x},~\forall j \in [N],~ \forall k \in [j+1,N-1],\hspace{-4pt}
        \label{eq:fastsls_qp_x_init}
    \end{align}
\end{subequations}
where $\mathcal{Q}_{k,j}, \mathcal{Q}_{N,j}$ are defined according to \cite[Eq. 20, 23]{leeman2024fast} (see App. \ref{app:sls_consistency}) for consistency between \eqref{eq:nmpc_problem_tightened} and \eqref{eq:fastsls_qp}. Notably, the right-multiplication by $\bar E_j$ prevents the cost in \eqref{eq:fastsls_qp_cost} from being expressed as a standard LQR objective in $\Phit^x$ and $\Phit^u$. Thus, it is not directly solvable using the Riccati-based methods of \cite{leeman2024fast, fang2026safe}. Instead, we solve \eqref{eq:fastsls_qp_reform}, which has an LQR cost and can be solved via the efficient solvers \cite{leeman2024fast, fang2026safe}:
\begin{subequations} \label{eq:fastsls_qp_reform}
    \begin{align}
        \hspace{-8pt}\min_{\Phix, \Phiu} \ 
        &\sum_{j=0}^{N-1}\hspace{-3pt}
        \sum_{k=j}^{N-1} \| \mathcal{Q}_{k,j} \Phit_{k,j}\|^2_{\mathcal{F}} + \|\mathcal{Q}_{N,j} 
         \Phix_{N,j}\|^2_{\mathcal{F}}\label{eq:fastsls_qp_cost_reform} \\
        \text{s.t.}\ \ & \Phix_{k+1,j} = A_k \Phix_{k,j} + B_k \Phiu_{k,j}, \label{eq:fastsls_qp_err_dyn}\\
        & \Phix_{j+1,j} = I_{\nx},~ \forall j \in [N],~\forall k \in [j+1,N-1].
        \label{eq:fastsls_qp_cost_reform_init}
    \end{align}
\end{subequations}
Note that \eqref{eq:fastsls_qp_reform} coincides with \eqref{eq:fastsls_qp} if $\bar E_j = I$ for all $j \in [N]$. In Prop. \ref{prop:sls}, we prove that the set of minimizers of \eqref{eq:fastsls_qp} and \eqref{eq:fastsls_qp_reform} are equivalent if $\bar E_j$ is right-invertible for all $j \in [N]$ (see App. \ref{app:sls} for the proof).

\begin{proposition}\label{prop:sls}
Let $(\TPhix,\TPhiu)$ be an optimizer of \eqref{eq:fastsls_qp}, and let $(\HPhix,\HPhiu)$ be an optimizer of \eqref{eq:fastsls_qp_reform}. If $\bar E_j$ is right-invertible for all $j \in [N]$, i.e., there exists $\bar E_j^{\ddagger}$ such that $\bar E_j \bar E_j^{\ddagger} = I_{\nx}$, then the two optimizers coincide: $\TPhix = \HPhix, \TPhiu = \HPhiu$.
\end{proposition}

When implementing the method of Sec. \ref{sec:method}, we add a small positive diagonal padding to $\bar R_k$ so that it is always invertible, and thus $\bar E_j$ \eqref{eq:zonotope_combination} is always right-invertible for all $j \in [N]$. 
After solving \eqref{eq:fastsls_qp_reform}, we compute the constraint tightening terms in \eqref{eq:tight_tube_center}--\eqref{eq:tight_terminal_constraint}. This updates the reachable tube centers $\bm{\mu}_k$ and widths $\bm{\sigma}_k$: 
$\bm{\sigma}_k \gets h_k^{\sigma}(\mathbf{\Phi}, \bm{\mathcal{D}})~\text{and}~\bm{\mu}_k \gets h_k^{\mu}(\mathbf{\Phi},\bm{\mathcal{D}}),
$
\looseness-1which are used as constraint tightenings in the next iteration of nominal trajectory optimization \eqref{eq:nmpc_problem_tightened}. This iterates until convergence (as in \cite[Alg. 1]{leeman2024fast}); we terminate after a fixed number of iterations or if the change in the iterates is below a threshold (see App. \ref{app:algorithm_block} for an algorithm block).

\noindent\textbf{Formulation Novelty:} Here, we summarize the changes made to prior nonlinear SLS formulations \cite[App. A]{leeman2024fast} \cite{leeman2025robust}:
\begin{enumerate}[leftmargin=*, itemindent=0pt]
    \item \looseness-1Zonotopes enable exact tightening computation \eqref{eq:tight_tube_center}--\eqref{eq:tight_terminal_constraint} by concatenating the generators of the exogenous disturbance and linearization error, unlike prior ellipsoidal methods \cite[App. A]{leeman2024fast}, but this yields right-invertible $\bar E_j$ incompatible with existing solvers \cite{leeman2024fast, fang2026safe}. We address this by setting $\Phix_{j+1,j} = I_{n_x}$ and including the disturbance scaling directly into the tightenings and cost \eqref{eq:surrogate_cost_h0}, with a compatible reformulation \eqref{eq:fastsls_qp_reform} justified by Prop. \ref{prop:sls}.
    \item Finally, we extend \cite{leeman2025robust, leeman2024fast} to handle non-zero-centered disturbances, enabling direct use of interval bounds from \eqref{eq:zono_map} without extra conservativeness due to zero-centering.
\end{enumerate}

\subsection{Efficient SLS Implementation on the GPU} \label{sec:gpu_implementation}
We propose GPUSLS-Linearization Error Optimization (GPUSLS-LEO), which builds on GPUSLS \cite{fang2026safe} to explicitly account for linearization error propagation. 
We first replace the ellipsoidal disturbance representation of \cite{fang2026safe} with a zonotope, modifying the constraint tightenings and cost terms via \eqref{eq:tight_constraint}-\eqref{eq:surrogate_cost_h0}. At each controller update, we modify GPUSLS by including a zonotopic bound on the linearization error via \eqref{eq:zonotope_combination}. We then modify the initialization of $\mathbf{\Phi}^\text{x}$ via \eqref{eq:fastsls_qp_cost_reform_init}. 

All other steps closely follow the GPUSLS procedure. In particular, GPUSLS-LEO solves \eqref{eq:fastsls_problem} iteratively via sequential quadratic programming (SQP). At each iteration, GPUSLS-LEO quadraticizes the cost and linearizes $x_{k+1} = f(x_k, u_k)$ around the current $(\textbf{z}, \textbf{v})$ to obtain $A_k, B_k$ \eqref{eq:dynamics_ltv}. It then solves (A) a local QP approximation of \eqref{eq:nmpc_problem_tightened} to update $(\textbf{z}, \textbf{v})$, (B) \eqref{eq:fastsls_qp_reform} to get $\Phit$, and updates (C) the tightenings via \eqref{eq:rmpc_tube_width_center}. Since our LEBs are differentiable, we can penalize error accumulation during step (A) by modifying the nominal objective \eqref{eq:nmpc_problem_tightened_cost} to $\bar J(\textbf{z}, \textbf{v}):=J(\textbf{z}, \textbf{v})+\lambda J_\textrm{err}(\textbf{z}, \textbf{v})$, where
\begin{equation}\label{eq:linearization_error_cost_penalty}
    \hspace{-2pt}J_\textrm{err}(\textbf{z}, \textbf{v}) = \sum_{j=0} ^ {N - 1}\sum_{k=j}^{N - 1} \| \mathbf{\Phi}_{k,j} \mathcal{Z}_r(z_j, v_j, \bm{\mu}_j, \bm{\sigma}_j)\|_\mathcal{F}^2,
\end{equation} 
 penalizes the propagated residual error \eqref{eq:linearization_error_all}. Here, $\mathbf{\Phi}, \bm{\mu}_j, \bm{\sigma}_j$ are from the previous iteration and $\lambda$ is a weighting term. Notably, \eqref{eq:linearization_error_cost_penalty} biases $(\textbf{z}, \textbf{v})$ toward regions of low propagated exogenous and linearization error disturbances.
Finally, we show that by solving \eqref{eq:fastsls_problem}, GPUSLS-LEO ensures containment in the reachable tubes \eqref{eq:sls_reach_non_center}, i.e., $x_k \in \bar\Omega_k^x$ and $u_k \in \bar \Omega_k^u$, for all $k \in [N]$. Formally, we have (proof in App. \ref{app:proofs}):
\begin{theorem}
Let $\mathbf{z}$, $\mathbf{v}$, and $\Phit$ be a feasible solution for \eqref{eq:fastsls_problem}. Then, \eqref{eq:sls_reach_non_center} is guaranteed to overapproximate the true closed-loop reachable set \eqref{eq:reach}, i.e., $\Omega_k \subseteq \bar\Omega_k$ for all $k \in [N]$.
\end{theorem}

\section{Results}

\looseness-1We evaluate LEB tightness (Sec.~\ref{sec:linearization_error_results}) and the robust controllers given by GPUSLS-LEO (Sec.~\ref{sec:gpusls_leo_results}). Runtime tests use an NVIDIA RTX 4090, except the multi-quadrotor and long-horizon experiments, which require an NVIDIA H200. 
System definitions appear in App.~\ref{app:dynamic_models}.
\subsection{Linearization Error} \label{sec:linearization_error_results}
\looseness-1We compare the LEB from our path-based Hessian (\textbf{PBH}) method on the \textit{Satellite} ($\nx=7$, $\nuu=3$) and \textit{Quadrotor} ($\nx=12$, $\nuu=4$) systems (App. \ref{app:dynamic_models}) with several baselines. \textbf{Global sampling}~\cite{leeman2025robust} estimates worst-case curvature \textit{offline} (taking over 7 minutes \cite{leeman2025robust}) over the \emph{global} set $\mathcal{X}\times\mathcal{U}$, yielding an \textit{empirical} over-approximation. \textbf{Interval Hessian (IH) (classic)}~\cite{leister2025robust, althoff2008reachability} evaluates~\eqref{eq:interval_quadratic_form} on a \textit{local} set $\mathcal{T}\subset(\mathcal{X}\times\mathcal{U})$. \textbf{IH (classic)-CORA} denotes the standard IH implementation in CORA~\cite{leister2025robust,Althoff2015ARCH}. \textbf{Random sampling} uses 10K samples from the local set and gives an \emph{under-approximation}.
We sample nominal points $(z,v)$ from $\mathcal{X}\times\mathcal{U}$ and define local sets $\mathcal{T}:=(z,v)\oplus \epsilon\mathcal{E}_{\nx+\nuu}$ with $\epsilon=0.01,0.05,0.1,0.15,0.2$. Fig.~\ref{fig:Linearization_error} shows the mean interval width of \eqref{eq:zono_map}. For \emph{Satellite}, where nonlinearity is mild, all over-approximation methods remain close to the random-sampling lower bound as $\epsilon$ grows, though PBH is still the tightest. For \emph{Quadrotor}, PBH consistently gives the tightest bounds. By contrast, global sampling is loose at small $\epsilon$, suggesting global curvature misses local behavior, while the interval Hessian becomes overly conservative at $\epsilon=0.20$. Table~\ref{tab:runtime_comparison} compares \textit{online} runtime. PBH runs in under 1 ms, enabling real-time use. Although slower than the global bound, which is efficiently calculated online after slow offline computation, PBH is 60\% and 27\% (satellite and quadrotor respectively) faster than dense random sampling and negligible running time compared to CORA IH.

\begin{figure}[h]
    \centering
    \includegraphics[width=\linewidth]{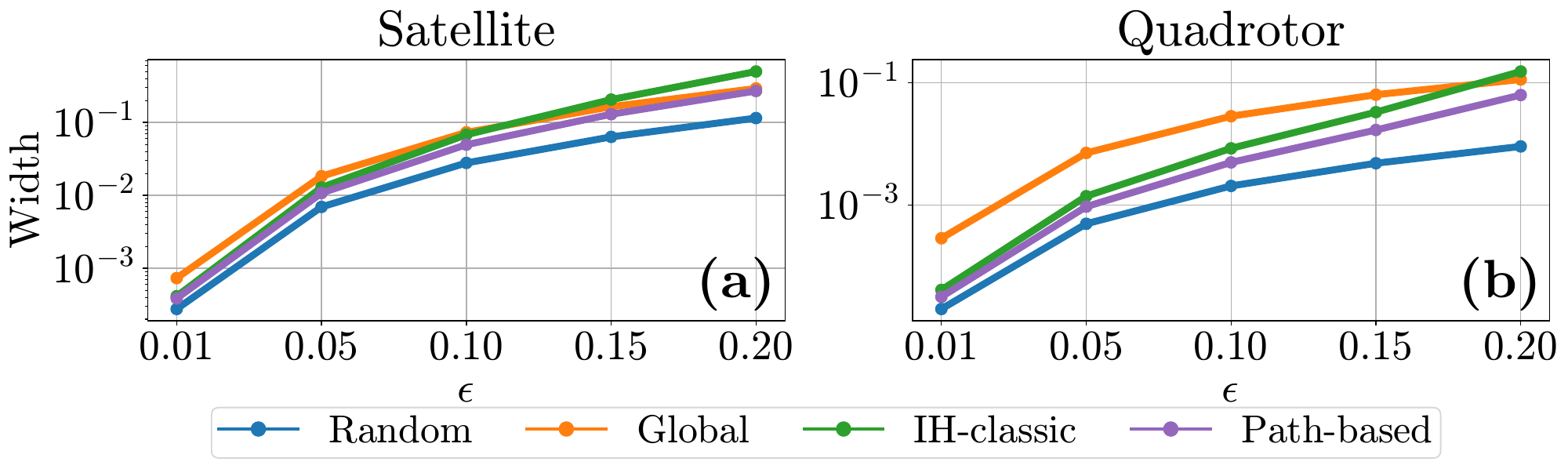}
    \caption{\textbf{(a)}: LEBs on the satellite system. \textbf{(b)}: LEBs on the quadrotor system. In both systems, our method achieves the tightest over-approximation. 
    }
    \label{fig:Linearization_error}
\end{figure}

\begin{table}[h]
\centering
\scriptsize
\caption{Runtime (in ms) comparison of linearization bound methods.}
\label{tab:runtime_comparison}
\setlength{\tabcolsep}{4pt}
\begin{tabular}{lccccc}
\toprule
System & Random & Global & IH-classic & IH-classic-CORA & Path-based \\
\midrule
Satellite & 0.894 & 0.143 & 0.218 & 11570 & 0.361 \\
Quadrotor & 1.278 & 0.156 & 0.453 & 227200 & 0.936 \\
\bottomrule
\end{tabular}
\end{table}

\begin{table}[h]
\centering
\scriptsize
\caption{Linearization error under nominal optimization via PGD.}
\label{tab:nominal_opt}
\setlength{\tabcolsep}{4pt}
\begin{tabular}{lcccc}
\toprule
System & Initial & Sampling & PGD & PGD better (\%) \\
\midrule
Satellite & 0.2680 & 0.2365 & 0.2241 & 100 \\
Quadrotor & 0.02253 & 0.01639& 0.01590  & 100 \\
\bottomrule
\end{tabular}
\end{table}
To quantify the benefit of differentiability in our path-based Hessian implementation, we search for a nominal point within a fixed perturbation radius $\epsilon$ that minimizes the LEB size. Starting from random nominal points, we run projected gradient descent (PGD) within a local neighborhood of radius $\epsilon_{\text{opt}}$ to reduce the error interval width. We compare PGD with random sampling of nominal points under the same time budget, where both methods use the PBH to evaluate linearization error. Table~\ref{tab:nominal_opt} shows that across 100\% of test cases, PGD yields smaller error interval widths than  random sampling across both systems, showing the value of gradient-based optimization for finding nominal points with tighter LEBs, which can benefit downstream SLS. 

\subsection{Using Linearization Bounds for RNOCP} \label{sec:gpusls_leo_results}

\looseness-1In this section, we evaluate \textbf{GPUSLS-LEO} on a suite of analytic and neural dynamical systems against various RNOCP baselines. We further assess the impact of key ablations, including the effect of optimizing over the linearization error gradients and the use of nonzero-centered zonotopes. All benchmarked methods are evaluated under the same disturbance set for the given dynamical system and, where applicable, are run for a maximum of 100 SQP iterations.

\begin{figure}[t]
    \centering
    \includegraphics[width=\linewidth]{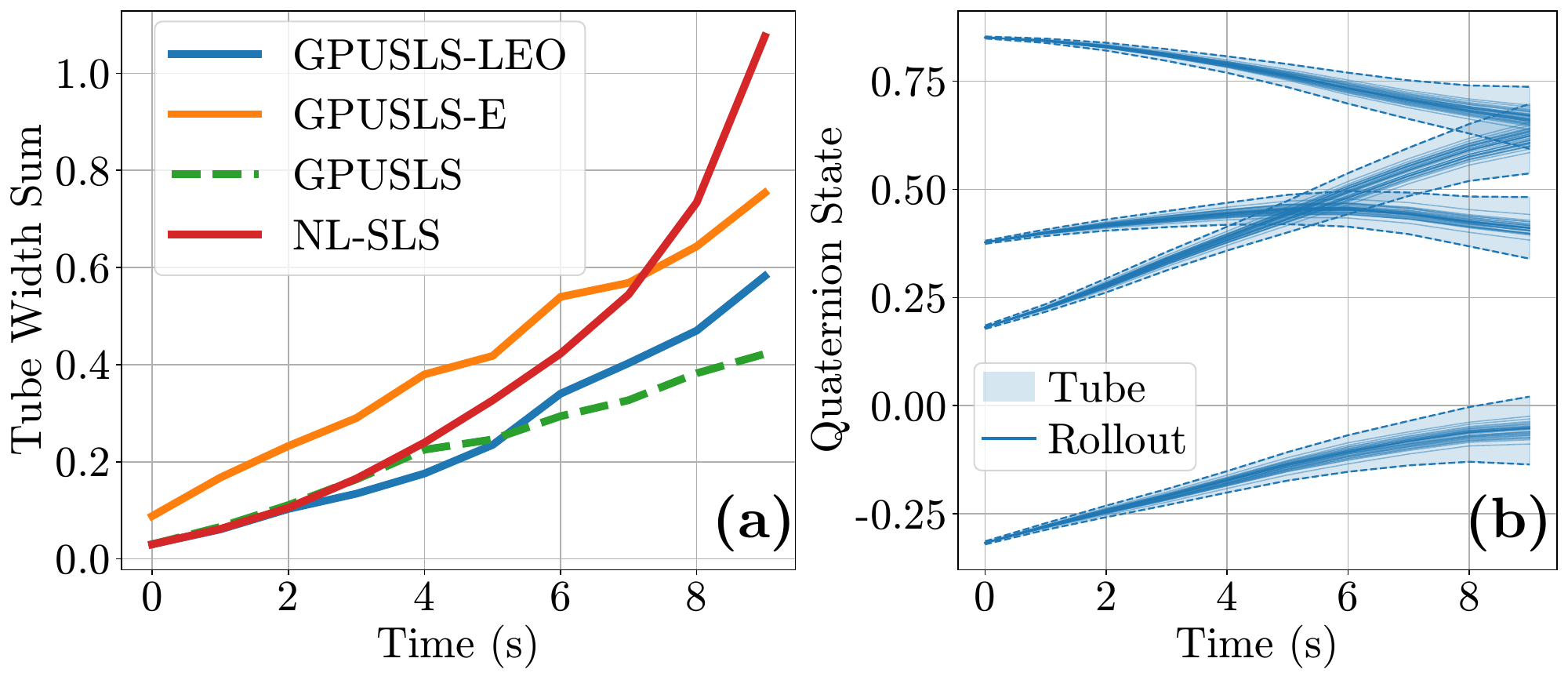}
    \caption{\looseness-1\textbf{Satellite (7D)}. \textbf{(a)}: Tube sizes for our method, GPUSLS-LEO, compared to baselines and the variant without linearization error. Our method is the least conservative, while minimally increasing tube size \textbf{(b)}: Tubes for the first four state dimensions; our tubes capture 100\% of sampled rollouts.
    }
    \label{fig:sat_tubes}
\end{figure}

\subsubsection{Satellite} \looseness-1We evaluate our method on a 7D satellite attitude control task, where the objective is to stabilize the system to a desired orientation under bounded disturbances. We adopt the experimental setup of \cite{leeman2025robust} and benchmark against their approach. In Fig.~\ref{fig:sat_tubes}\,\textbf{(a)} we compare the tube widths $\bm{\sigma}_k$ for: \textbf{GPUSLS-LEO}, \textbf{GPUSLS-E} (ellipsoidal propagation as in \cite[App. A]{leeman2024fast}), \textbf{GPUSLS} (no linearization error propagation) \cite{fang2026safe}, and \textbf{NL-SLS} (\hspace{1pt}\cite{leeman2025robust}). Our method achieves the tightest tubes among all baselines. We have on average 20\% smaller tubes than NL-SLS, which relies on a global maximum Hessian bound and is therefore locally overly-conservative along a trajectory. Compared to GPUSLS-E, our method achieves 45\% smaller tubes, as it involves overapproximating the Minkowski sum of ellipsoids, whereas our zonotope-based combination is exact. We observe that incorporating our LEB pipeline does not introduce significant conservativeness, as the tube widths increase by only 4\% compared to GPUSLS. In Fig.~\ref{fig:sat_tubes}\,\textbf{(b)}, we show the rollouts of 20 random and 64 adversarial disturbances for the quaternion states, showing 100\% containment under our method. 

\begin{figure}[t]
    \centering
    \includegraphics[width=\linewidth]{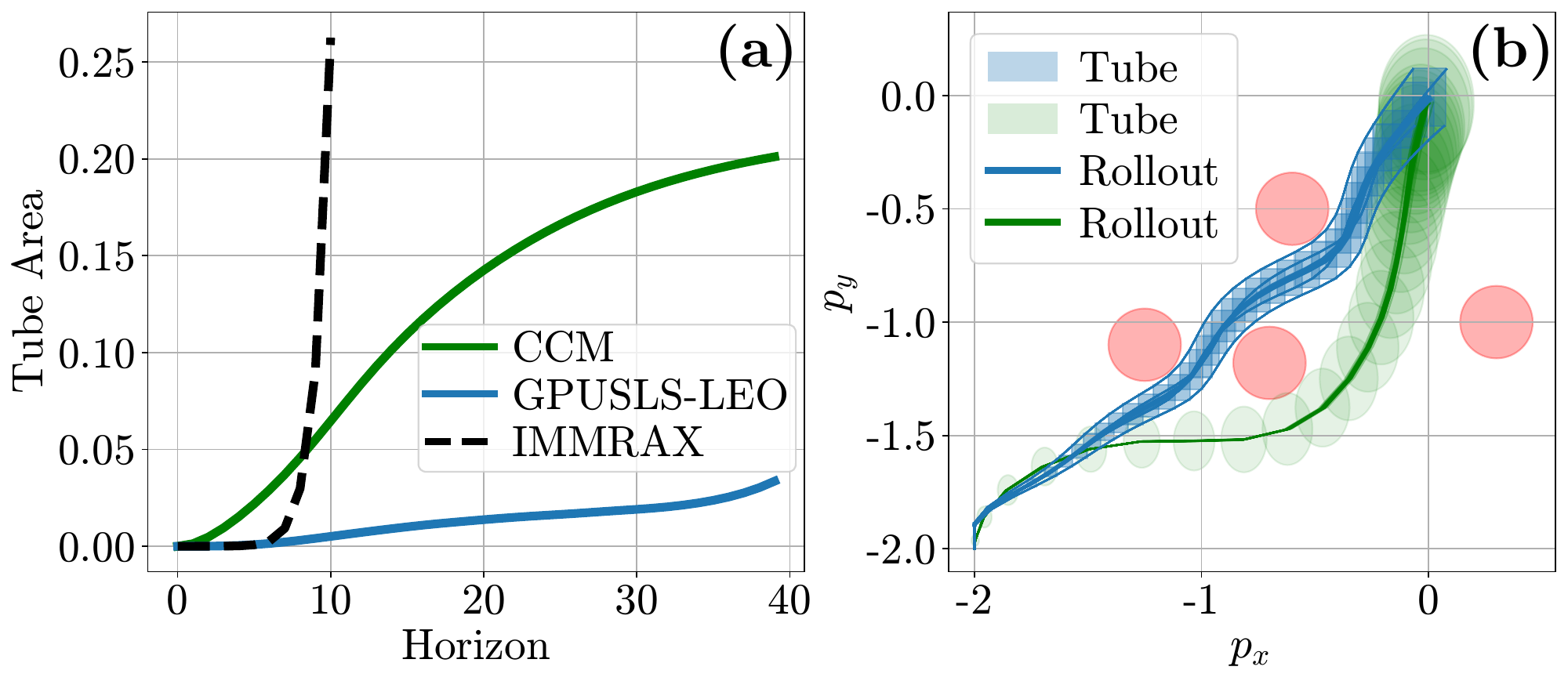} \vspace{-10pt}
    \caption{\looseness-1\textbf{Planar quadrotor (6D)}. \textbf{(a)}: Tube area comparison against CCM and \texttt{immrax}, showing our method is less conservative. \textbf{(b)}: Tubes and rollouts for GPUSLS-LEO and CCM. Both methods achieve 100\% safety, but the conservativeness of the CCM forces the system to take a suboptimal path.
    }
    \label{fig:planar_tubes}
\end{figure}

\begin{figure}[t]
    \centering
    \includegraphics[width=\linewidth]{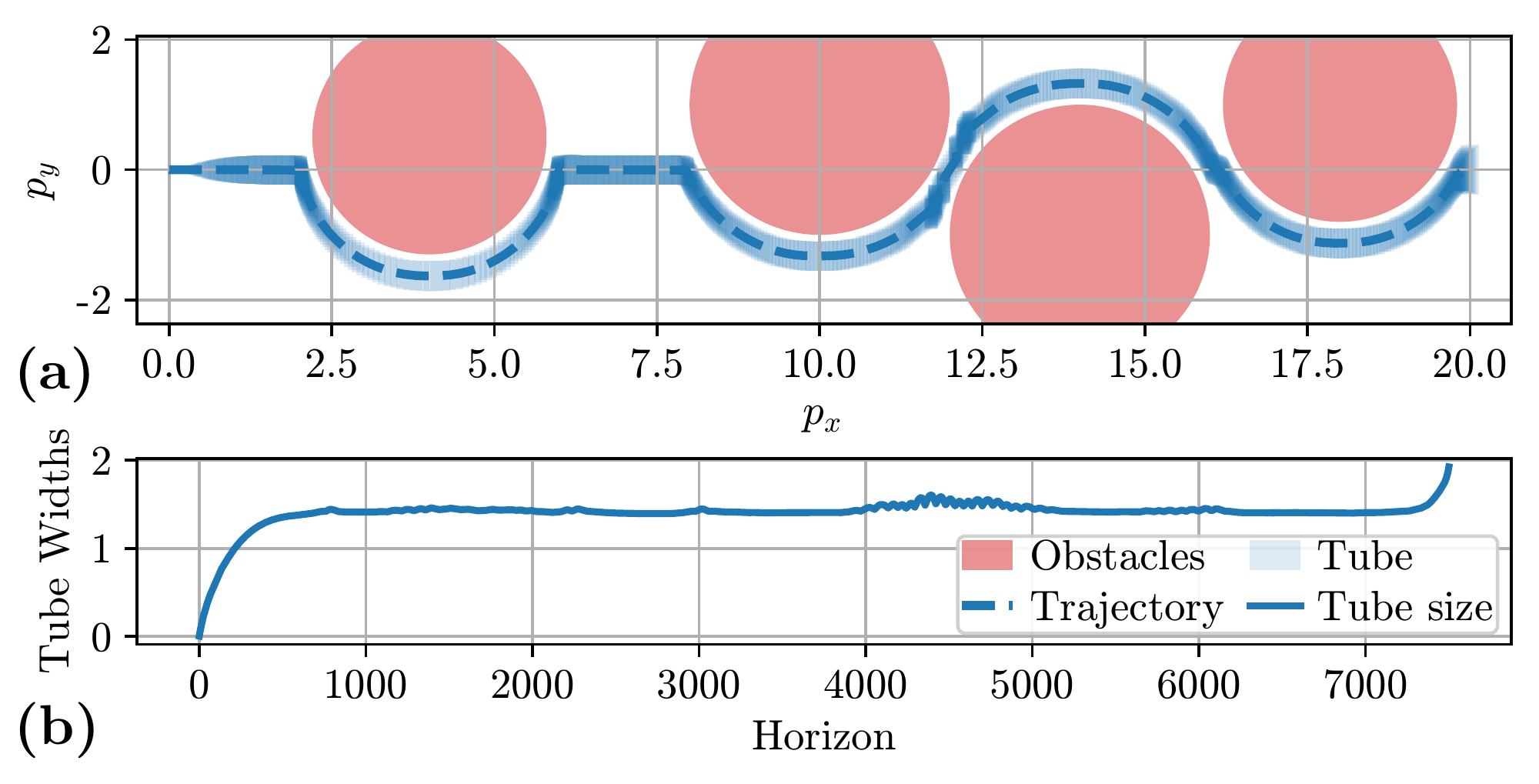}
    \caption{\looseness-1\textbf{Planar quadrotor (6D).} \textbf{(a):} System navigating through an obstacle field over a 20 m trajectory with a planning horizon length of 7500. \textbf{(b):} Tube widths along the horizon. Despite the long horizon, the tubes remain bounded due to the GPUSLS-LEO controller's optimization of tube sizes.} 
    \label{fig:long_tubes}
\end{figure} 

\subsubsection{Planar quadrotor} 

\looseness-1We evaluate our method on a 6D planar quadrotor against control contraction metrics (CCM)~\cite{sasfi2023robust}. Both methods must navigate through a dense obstacle field over a horizon of 40 steps. Fig.~\ref{fig:planar_tubes}\textbf{(a)} shows our method achieving 91\% smaller tube area than the CCM for the $(x,y)$ position states. This arises because the CCM certifies safety over the entire state space, introducing inherent conservativeness. GPUSLS-LEO, on the other hand, certifies safety along a specific trajectory, enabling tighter, trajectory-dependent bounds. We also evaluate the closed-loop system on an interval-based reachability analysis tool \texttt{immrax} \cite{harapanahalli2024immrax}, showing its calculated tube area diverging after a few steps.

\looseness-1In Fig.~\ref{fig:planar_tubes}\textbf{(b)} we compare the calculated trajectories and tubes for our method and CCM. Our approach is less conservative, allowing a more direct, task-optimal path, while the CCM is forced to take a longer route. In Fig.~\ref{fig:long_tubes}\textbf{(a)}, we show the trajectory and tubes over a long horizon of 7500 steps, which involves optimizing a state-control trajectory with $1.95\times 10^5$ decision variables. Due to the tightness of our LEBs and the optimization of the closed-loop controller in SLS, our tubes remain well-bounded and finite over the horizon (Fig.~\ref{fig:long_tubes}\textbf{(b)}), demonstrating its scalability to long horizons.

\begin{figure}[t]
    \centering
    \includegraphics[width=\linewidth]{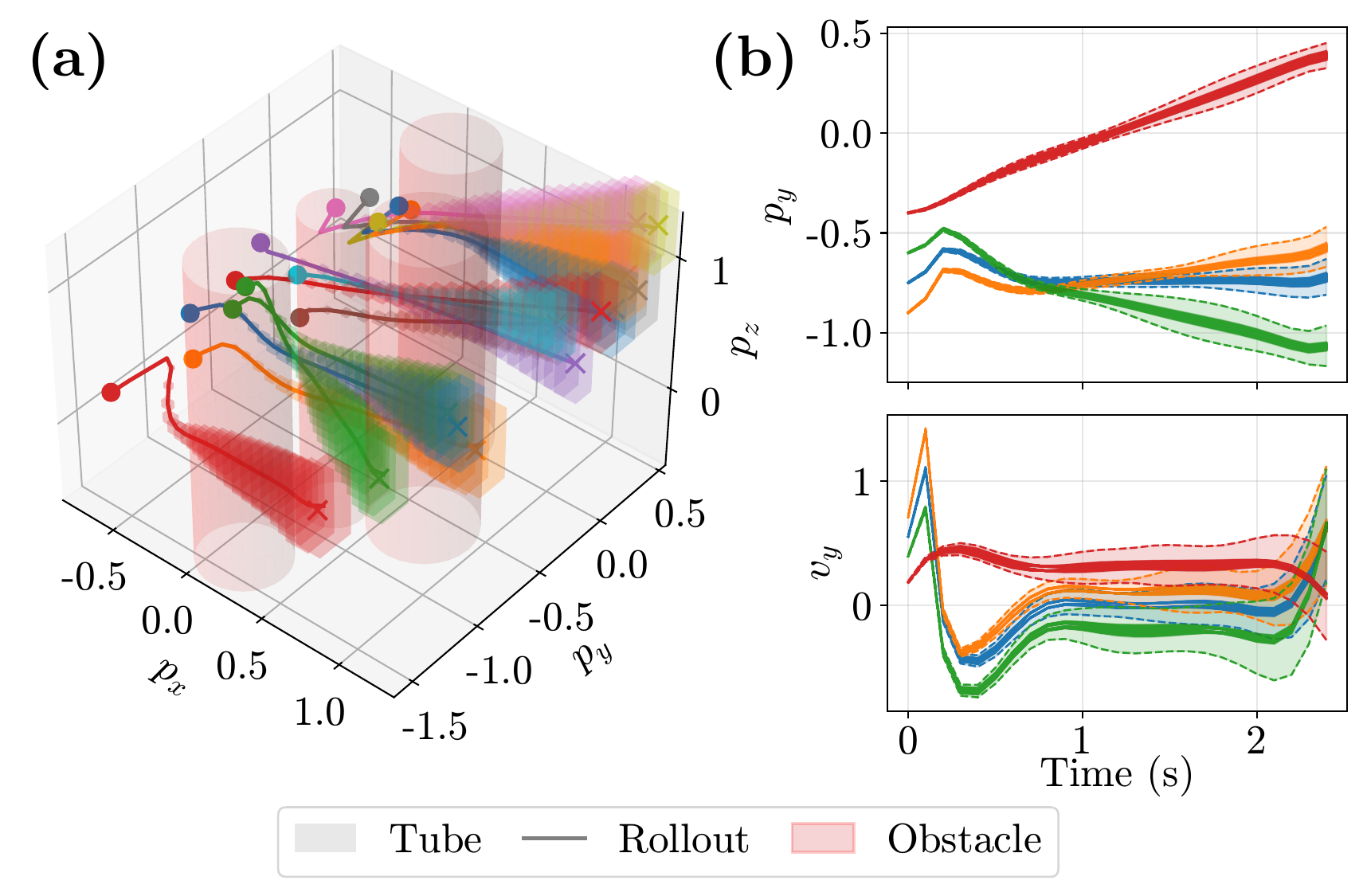}
    \caption{\textbf{System of 14 coupled 3D quadrotors (168D)}. \textbf{(a)}: Computed trajectories and tubes in an obstacle field. Even at high dimensions, our method is able to compute tight tubes. \textbf{(b)}: Tubes and 100 rollouts for states $p_y$ and $v_y$ for 4 quadrotors. 100\% of rollouts stay safely in the tubes.
    } 
    \label{fig:coupled_quadrotors}
\end{figure}

\subsubsection{Coupled 3D Quadrotor}
To demonstrate scalability to high dimensional systems, we evaluate our method on a coupled system of 14 3D quadrotors (168D) over a horizon of 25. The system is tasked with navigating through an obstacle field while the quadrotors are coupled via a spring force toward the group centroid. We show the planned trajectories and the tubes in Fig.~\ref{fig:coupled_quadrotors}\textbf{(a)}, and observe all rollouts remain in the tight tubes despite the high dimensional state space, demonstrating the method's scalability to state dimensionality. In Fig.~\ref{fig:coupled_quadrotors}\textbf{(b)}, we plot the tubes and the corresponding rollouts for $p_y$ and $v_y$ for 100 adversarial disturbances.

\begin{figure}[t]
    \centering
    \includegraphics[width=\linewidth]{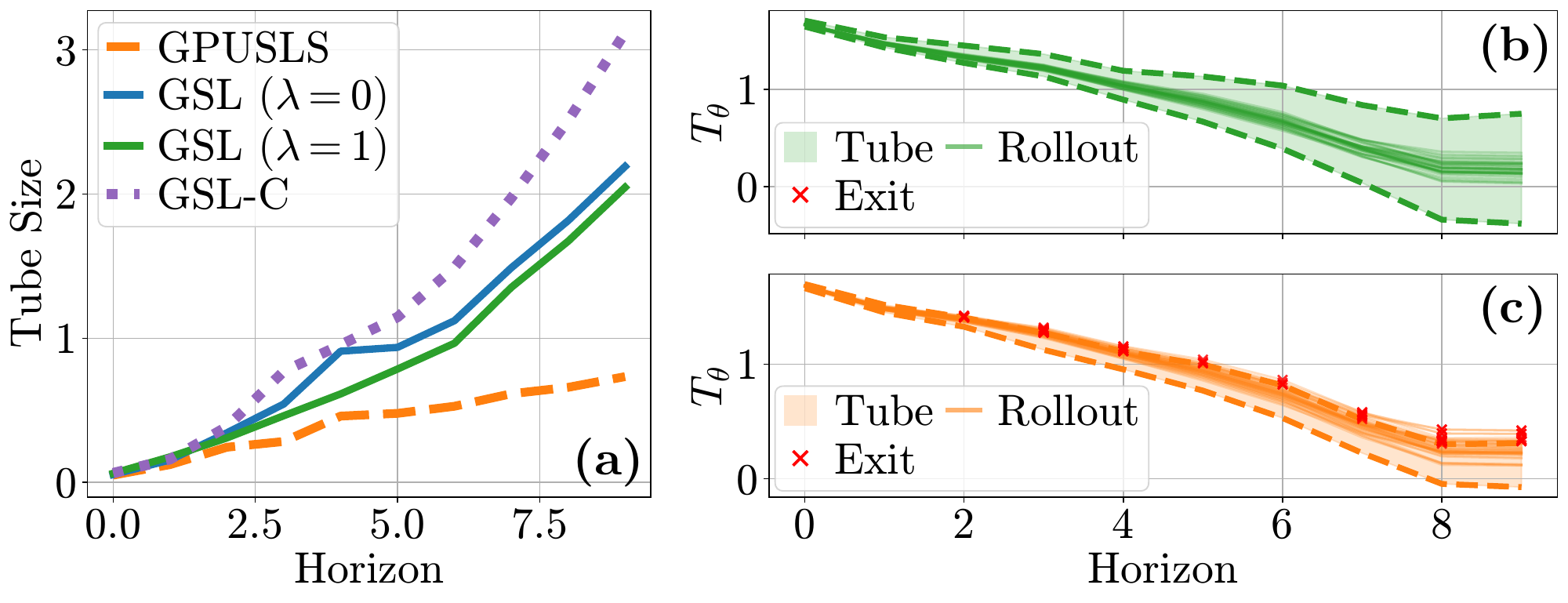}
    \caption{\textbf{Neural T-pusher (5D).} \textbf{(a)}: Tube size comparison across ablations of our method, showing nonzero-centered tubes and linearization error gradients reduce tube widths. \textbf{(b)}: Rollouts of our method; all trajectories remain in the computed robust tubes. \textbf{(c)}: Rollouts of GPUSLS. Without formally considering linearization error, rollouts leave the tube.
    }
    \label{fig:neural_tubes}
\end{figure}

\subsubsection{Neural T-pusher}
We demonstrate that our formulation can account for linearization error propagation through learned neural dynamics. We train an MLP-based model for the T-pusher system and task it with rotating and pushing the object to the goal configuration (Fig.~\ref{fig:neuralt_rollouts}). We evaluate our method on several ablations (Fig.~\ref{fig:neural_tubes}\textbf{(a)}), including \textbf{GPUSLS}, \textbf{GSL} (shorthand for GPUSLS-LEO), \textbf{GSL-C} (zero-centered zonotopes) and \textbf{GSL ($\lambda = 1$)}, which incorporates a penalty on the linearization error as defined in \ref{sec:gpu_implementation}. Compared to GPUSLS, GPUSLS-LEO produces tubes that are on average 99\% larger. However, Fig.~\ref{fig:neural_tubes}\textbf{(b)} shows that our method achieves 100\% containment of all rollouts, while Fig.~\ref{fig:neural_tubes}\textbf{(c)} illustrates how trajectories from GPUSLS fall out of the tube, showing the need to formally consider linearization error. 

Fig.~\ref{fig:neural_tubes}\textbf{(a)} shows that adding linearization error penalties can reduce conservativeness, with $\lambda = 1$ decreasing tube widths by 8\%. 
We also show the importance of nonzero-centered zonotopes, as using zero-centered zonotopes leads to 25\% larger tubes. We present rollouts and disturbance deviations in Fig.~\ref{fig:neuralt_rollouts}\textbf{(a)}. 
In Fig.~\ref{fig:neuralt_rollouts}\textbf{(b)}, we demonstrate the method's ability for closed-loop MPC on simulated contact dynamics, where the pusher successfully completes the task. We note that we do not formally consider learning error; however, it can be incorporated as in \cite{srinivasan2026safety, nath2026pixels}.

\begin{table}[h]
\centering
\scriptsize
\caption{Per iteration runtime breakdown (ms) across systems.}
\label{tab:runtime_breakdown}
\setlength{\tabcolsep}{3pt}
\begin{tabular}{lccccc}
\toprule
System & GPUSLS & Lin. Error Grad. & Remainder & Total & DT \\
\midrule
Satellite           & 11  & 3   & $<1$ & 15  & 1000 \\
Planar Quadrotor   & 18  & 1   & $<1$ & 20  & 150  \\
Coupled Quadrotor (96D) & 40  & --  & 48   & 88  & 100  \\
Neural T-Dynamics  & 630 & 100 & 32   & 762 & 1000 \\
\bottomrule
\end{tabular}
\end{table}

\begin{figure}[t]
    \centering
    \includegraphics[width=\linewidth]{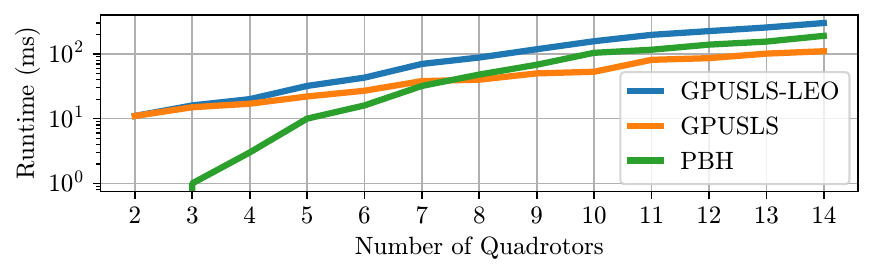}
    \caption{\textbf{Coupled Quadrotor.} Per-iteration runtime evaluation across increasing numbers of coupled quadrotors. GPUSLS scales favorably with state dimension, growing approximately logarithmically with state and control space. Path-Based Hessians (PBH) achieves low runtimes at small dimensions but scales more poorly as the state dimension increases. 
    }
    \label{fig:coupled_runtime_scaling}
\end{figure} 

\subsubsection{Runtime} \looseness-1We evaluate the runtime of our method across all systems. In Tab.~\ref{tab:runtime_breakdown}, we show the per iteration breakdown in milliseconds of the different components of our methodology. In all systems aside from the coupled quadrotor, the main computation load is in the GPUSLS, showing that our proposed modifications do not drastically increase run time. We also show that all total runtimes fall below the discretization time step DT, demonstrating the viability of our method under a real-time iteration MPC scheme \cite{fang2026safe}.

\looseness-1Finally, we evaluate the method's scaling with growing state and control dimensions using an increasing number of coupled quadrotors. Fig.~\ref{fig:coupled_runtime_scaling} shows that GPUSLS incurs higher initial runtime but scales favorably with state and control dimension due to its logarithmic scaling \cite{fang2026safe}. PBH on the other hand, exhibits low initial runtime, but can slow down in higher dimensions due to the cost of Hessian computations.

\section{Conclusion}
This paper presents a framework for computing tight, differentiable, and sound linearization error bounds (LEBs) for nonlinear analytic and NN dynamics and for using them in GPUSLS-LEO, a real-time GPU-parallel RNOCP solver based on SLS. The approach reduces conservativeness, outperforming state-of-the-art reachability analysis tools, while maintaining formal guarantees and achieving performance on systems with up to 168 states, real-time control rates (up to 67 Hz), and problems with $\approx 2\times 10^5$ decision variables. Future work will investigate how our LEBs can be extended to bound linearization errors for contact-rich planning \cite{li2026certified} and nonlinear output-feedback control, including through extensions of the SLS-based methods described in \cite{leeman2026vision, zhou2023safe, li2026robustness}.

\bibliographystyle{IEEEtran}
\bibliography{references_trunc}

\clearpage
\appendix
\appendices

In the appendix, we first discuss details on the system dynamics used in the experimental results (App. \ref{app:dynamic_models}). We then provide an algorithm block describing how sequential quadratic programming is used to solve GPUSLS-LEO (App. \ref{app:algorithm_block}). Next, we discuss and prove Proposition \ref{prop:sls} and a supporting lemma in App. \ref{app:sls}. We then provide proofs for the remaining theoretical results that were omitted from the main manuscript in App. \ref{app:proofs}. Finally, we discuss in more detail the Lagrange multiplier variable $\tau$ used in our framework in App. \ref{app:sls_consistency}.

\section{Dynamical System Definitions} \label{app:dynamic_models}
In this section, we give an overview of the discrete-time dynamical system definitions for each system. The discrete-time dynamics are obtained via forward Euler discretization with timestep $\Delta t$, $x_{k+1} = x_k + \Delta t (\dot{x}_k)$.

\subsection{Satellite}
We use the same satellite dynamics and experimental setup as \cite{leeman2025robust}. Specifically, we define the dynamics as,

\begin{subequations}
    \begin{align}
        \dot{x} &= \begin{bmatrix}
            \zeta (\omega)q \\
            I_S^{-1}(v - w \times (I_S\omega))
        \end{bmatrix} \\
        \zeta(\omega) &:= \frac{1}{2} \begin{bmatrix}
            0 & -\omega_1 & -\omega_2 & -\omega_3 \\
            \omega_1 & 0 & \omega_3 & -\omega_2 \\
            \omega_2 & -\omega_3 & 0 & \omega_1 \\
            \omega_3 & \omega_2 & -\omega_1 & 0
        \end{bmatrix}
    \end{align}
\end{subequations}
with states $x := (q, \omega) \in \mathbb{R}^7$, where $q \in \mathbb{R}^4$ is the attitude quaternion, $\omega \in \mathbb{R}^3$ the angular rotation rate, $v \in \mathbb{R}^3$ the input control torque, and $I_S = \mathrm{diag}(5, 2, 1)$ the symmetric inertia matrix. We define a constant disturbance scaling of $E = 5 \cdot 10^{-3} \cdot \mathrm{diag}(0, 0, 0, 0, 1, 1, 1)$.

\subsection{Planar Quadrotor} 
We consider a planar quadrotor with the following dynamics
\begin{subequations}
\begin{align}
\dot{x} &=
\begin{bmatrix}
v_x \\
v_y \\
\dot{\phi} \\
-\frac{1}{m}(u_1 + u_2)\sin(\phi) \\
\frac{1}{m}(u_1 + u_2)\cos(\phi) - g \\
\frac{L}{J}(u_2 - u_1)
\end{bmatrix}
\end{align}
\end{subequations}
The state is defined as $x := (p_x, p_y, \phi, v_x, v_y, \dot{\phi}) \in \mathbb{R}^6$ and input $u := (u_1, u_2) \in \mathbb{R}^2$. The position is denoted as $(p_x, p_y)$, $\phi$ the pitch angle, $(v_x, v_y)$ the translational velocities, and $\dot{\phi}$ the angular velocity. The inputs $(u_1, u_2)$ are the individual rotor thrusts. We define the mass $m = 2.0576$, gravitational acceleration $g = 9.81$, arm length $L = 0.25$ and moment of inertia $J = 0.01$. We define the disturbance scaling as $E = 5 \cdot 10^{-2} \cdot \text{diag}(0, 0, 0, 1, 1, 0)$.

\subsection{Coupled 3D Quadrotor}
We consider a system of $N$ quadrotors. For a given quadrotor $i$, we define its dynamics as
\begin{subequations}
\begin{align}
\dot{\xi}_i &=
\begin{bmatrix}
\dot{p}_{x,i} \\
\dot{p}_{y,i} \\
\dot{p}_{z,i} \\
\dot{\phi}_i \\
\dot{\theta}_i \\
\dot{\psi}_i \\
\dot{v}_{x,i} \\
\dot{v}_{y,i} \\
\dot{v}_{z,i} \\
\dot{p}_i \\
\dot{q}_i \\
\dot{r}_i
\end{bmatrix}
=
\begin{bmatrix}
v_{x,i} \\
v_{y,i} \\
v_{z,i} \\
p_i + (q_i \sin \phi_i + r_i \cos \phi_i)\tan \theta_i \\
q_i \cos \phi_i - r_i \sin \phi_i \\
\frac{q_i \sin \phi_i + r_i \cos \phi_i}{\cos \theta_i} \\
\frac{T_i}{m}\left(\cos \psi_i \sin \theta_i \cos \phi_i + \sin \psi_i \sin \phi_i\right) \\
\frac{T_i}{m}\left(\sin \psi_i \sin \theta_i \cos \phi_i - \cos \psi_i \sin \phi_i\right) \\
\frac{T_i}{m}\left(\cos \theta_i \cos \phi_i\right) - g \\
\frac{J_y - J_z}{J_x} q_i r_i + \frac{\tau_{\phi,i}}{J_x} \\
\frac{J_z - J_x}{J_y} p_i r_i + \frac{\tau_{\theta,i}}{J_y} \\
\frac{J_x - J_y}{J_z} p_i q_i + \frac{\tau_{\psi,i}}{J_z}
\end{bmatrix}.
\end{align}
\end{subequations}
The state of quadrotor $i$ is defined as $\xi_i := (p_{x,i}, p_{y,i}, p_{z,i}, \phi_i, \theta_i, \psi_i, v_{x,i}, v_{y,i}, v_{z,i}, p_i, q_i, r_i) \in \mathbb{R}^{12}$ where $(p_{x,i}, p_{y,i}, p_{z,i})$ denotes position,  $(\phi_i, \theta_i, \psi_i)$ the roll, pitch, and yaw angles,  $(v_{x,i}, v_{y,i}, v_{z,i})$ the translational velocities, and $(p_i, q_i, r_i)$ denote body angular velocities. The input is defined as $\upsilon_i := (T_i, \tau_{\phi,i}, \tau_{\theta,i}, \tau_{\psi,i}) \in \mathbb{R}^4$, where  $T_i$ is the collective thrust and $(\tau_{\phi,i}, \tau_{\theta,i}, \tau_{\psi,i})$ are the body torques.

For a system of $N$ coupled quadrotors, we stack the states and inputs, giving the overall state and control vectors, 
\begin{equation}
    x := (\xi_1, \ldots, \xi_N) \in \mathbb{R}^{12N}, \quad u := (\upsilon_1, \ldots, \upsilon_N) \in \mathbb{R}^{4N}.
\end{equation}
We additionally incorporate a centroid spring-damper coupling in the translational dynamics. We define the centroid position and velocity as
\begin{equation}
    \bar{p}_k = \frac{1}{N} \sum_{i=1}^{N} p_{i,k}, \quad \bar{v}_k = \frac{1}{N} \sum_{i=1}^N v_{i,k}.
\end{equation}
The coupling acceleration applied to quadrotor $i$ is
\begin{equation}
    a_{i,k}^\text{coup} = -\alpha N(p_{i,k} - \bar{p}_k) - \gamma (v_{i,k} - \bar{v}_k)
\end{equation}
where $\alpha$ and $\gamma$ are the coupling gains, with $\alpha$ specifying the stiffness of the spring term and $\gamma$ controlling the damping coefficient. In our experiments, we set $\alpha = 0.75$ and $\gamma = 0.25$. Therefore, the discrete-time dynamics become

\begin{equation}
\xi_{i,k+1}
=
\xi_{i,k}
+
\Delta t
\begin{bmatrix}
\dot{p}_{x,i,k} \\
\dot{p}_{y,i,k} \\
\dot{p}_{z,i,k} \\
\dot{\phi}_{i,k} \\
\dot{\theta}_{i,k} \\
\dot{\psi}_{i,k} \\
\dot{v}_{i,k} + a^{\mathrm{coup}}_{i,k} \\
\dot{p}_{i,k} \\
\dot{q}_{i,k} \\
\dot{r}_{i,k}
\end{bmatrix}.
\end{equation}

We use a disturbance scaling matrix of $E =  2 \cdot 10^{-2} \cdot I_{12N}$

\section{Algorithm Block}\label{app:algorithm_block}
\begin{algorithm}
\caption{GPUSLS-LEO}
\label{alg:gpusls_leo}
\begin{algorithmic}[1]
\Require $\bar x_0$, $J$, $f(x,u)$, $E(x)$, $g_k(x,u)$, $g^f(x)$
\While{\textsc{SQP not converged}}
    \State $\mathbf{A}$, $\mathbf{B}$, $\mathbf{G}$ $\leftarrow \textsc{SQP: Linearize System}$
    \State $\mathbf{Q}, \mathbf{M}, \mathbf{R}, \mathbf{q}, \mathbf{r} \leftarrow \textsc{SQP: Quadraticize Cost}$
    \While{\textsc{SLS not converged}}
        \State $\mathbf{z}$, $\mathbf{v}$ $\leftarrow \textsc{Optimize Nominal Trajectory}$
        \State $\bm{\sigma}, \bm{\mu} \leftarrow \textsc{Compute interval}$
        \State $\mathbf{\bar{E}}$, $\mathbf{\bar{c}}$ $\leftarrow \textsc{Compute residual error zonotope}$
        \State $\bm{\tau}, \bm{\mathcal{Q}} \leftarrow \textsc{Update duals and cost}$
        \State $\mathbf{\Phi}^\textrm{x}, \mathbf{\Phi}^\textrm{u} \leftarrow \textsc{Optimize Controller}$
        \State $h(\mathbf{\Phi}, \bm{\mathcal{D}}), h^f(\mathbf{\Phi}, \bm{\mathcal{D}}) \leftarrow \textsc{Update Tightenings}$
    \EndWhile
\EndWhile \\
\Return $\mathbf{z}, \mathbf{v}, \bm{\Phi}^x, \bm{\Phi}^u$
\end{algorithmic}
\end{algorithm}
\noindent Here, $\mathbf{Q}, \mathbf{M}, \mathbf{R}, \mathbf{q}, \mathbf{r}$ correspond to the quadratic approximation of the nominal cost. 
\section{On Solving for $\Phit$}\label{app:sls}

Solving \eqref{eq:fastsls_qp} yields optimal values for $\Phix$ and $\Phiu$ which penalize the size of the tightenings \eqref{eq:surrogate_cost_h0}. However, this leads to terms in \eqref{eq:fastsls_qp_cost} that scale the response matrices $\Phix$ and $\Phiu$ by $\bar E_j$ in the cost~\eqref{eq:fastsls_qp_cost}. This remains a quadratic cost in $\Phit$ but no longer maintains the structure of an LQR cost function. To address this, we instead solve \eqref{eq:fastsls_qp_reform} which has an LQR structure and can be solved via Riccati recursions. For self-containment, we restate the original problem \eqref{eq:fastsls_qp} and the reformulated problem \eqref{eq:fastsls_qp_reform}, and the equivalent optimizer Proposition (Prop. \ref{prop:sls}) below: 

\noindent\textbf{Original Problem:}
\begin{align}
    \hspace{-3pt}\min_{\Phix, \Phiu} \ 
    &\sum_{j=0}^{N-1} 
    \sum_{k=j}^{N-1} \| \mathcal{Q}_{k,j} \Phit_{k,j}\bar{E}_j \|^2_{\mathcal{F}} + \|\mathcal{Q}_{N,j} 
     \Phix_{N,j} \bar{E}_j\|^2_{\mathcal{F}}\tag{\ref{eq:fastsls_qp_cost}} \\
    \text{s.t.}\ \ & \Phix_{k+1,j} = A_k \Phix_{k,j} + B_k \Phiu_{k,j}, \tag{\ref{eq:fastsls_qp_x_prop}}\\
    & \Phix_{j+1,j} = I_{n_x},~\forall j \in [N],~ \forall k \in [j+1,N-1],\hspace{-4pt} 
    \tag{\ref{eq:fastsls_qp_x_init}}
\end{align}

\noindent\textbf{Reformulated Problem:}
\begin{align}
    \hspace{-5pt}\min_{\Phix, \Phiu} \ 
    &\textstyle\sum_{j=0}^{N-1}\hspace{-3pt}
    \sum_{k=j}^{N-1} \| \mathcal{Q}_{k,j} \Phit_{k,j}\|^2_{\mathcal{F}} + \|\mathcal{Q}_{N,j} 
     \Phix_{N,j}\|^2_{\mathcal{F}}\tag{\ref{eq:fastsls_qp_cost_reform}} \\
    \text{s.t.}\ \ & \Phix_{k+1,j} = A_k \Phix_{k,j} + B_k \Phiu_{k,j}, \tag{\ref{eq:fastsls_qp_err_dyn}} \\
    & \Phix_{j+1,j} = I_{\nx},~ \forall j \in [N],~\forall k \in [j+1,N-1]
    \tag{\ref{eq:fastsls_qp_cost_reform_init}}
\end{align}

In Prop. \ref{prop:sls}, we prove that the optimizer sets of \eqref{eq:fastsls_qp} and \eqref{eq:fastsls_qp_reform} are equivalent.
\begingroup
\renewcommand{\theproposition}{\ref{prop:sls}}
\begin{proposition}[Equivalent Optimizers]
Let $(\TPhix,\TPhiu)$ be an optimizer of \eqref{eq:fastsls_qp}, and let $(\HPhix,\HPhiu)$ be an optimizer of the reformulated problem \eqref{eq:fastsls_qp_reform}. Suppose that $\bar E_j$ is right-invertible for all $j \in [N]$, i.e., there exists $\bar E_j^{\ddagger}$ such that $\bar E_j \bar E_j^{\ddagger} = I_{\nx}$. Then the two optimizers coincide: $\TPhix = \HPhix, \TPhiu = \HPhiu$.
\end{proposition}
\addtocounter{proposition}{-1}
\endgroup

\begin{proof}
    Observe that \eqref{eq:fastsls_qp} can be rewritten into \eqref{eq:fastsls_qp_var} with the variable change $\Psit_{k,j} = \Phit_{k,j}\bar{E}_j$, since by using \eqref{eq:fastsls_qp_x_prop} and \eqref{eq:fastsls_qp_x_init} we have that, 
    \begin{eqnarray*}
        \Phix_{k+1,j}\bar{E}_j = A_k \Phix_{k,j}\bar{E}_j + B_k \Phiu_{k,j}\bar{E}_j &\iff \text{Eq. \eqref{eq:fastsls_qp_x_prop_var}},\\
    \text{and}~~~~~~~~~~~~~~~~~~~~~~~~~~~~~~~~~~~~~~~~~~~&\\
        \Phix_{j+1,j}\bar{E}_j = I_{n_x}\bar{E}_j &\iff \text{Eq. \eqref{eq:fastsls_qp_x_init_var}.}
    \end{eqnarray*}

    \begin{subequations} \label{eq:fastsls_qp_var}
    \begin{align}
        \hspace{-3pt}\min_{\Psix, \Psiu} \ 
        &\sum_{j=0}^{N-1}\hspace{-3pt}
        \sum_{k=j}^{N-1} \| \mathcal{Q}_{k,j} \Psit_{k,j}\|^2_{\mathcal{F}} + \|\mathcal{Q}_{N,j} 
         \Psix_{N,j}\|^2_{\mathcal{F}}\label{eq:fastsls_qp_cost_var} \\
        \text{s.t.} \ & \Psix_{k+1,j} = A_k \Psix_{k,j} + B_k \Psiu_{k,j},\label{eq:fastsls_qp_x_prop_var} \\
        & \Psix_{j+1,j} = \bar{E}_j,~\forall j \in [N],~\forall k \in [j+1,N-1]
        \label{eq:fastsls_qp_x_init_var}
    \end{align}
    \end{subequations}
    From \cite{fang2026safe}, \eqref{eq:fastsls_qp_var} can be solved via the parallel Riccati recursions,  
    \begin{equation} \label{eq:backward_prop}
    \begin{aligned}
        \mathcal{G}_{N,j} &= \mathcal{Q}^{\text{x}}_{N,j}, \quad \quad 
        \mathcal{K}_{k,j} = -\mathcal{G}_{k,j} \mathcal{B}_{k,j}, \\
        \mathcal{P}_{k,j} &= \mathcal{Q}^\text{x}_{k,j} + \big(A_k\big)^\top \mathcal{P}_{k+1,j}A_k + \big(\mathcal{K}_{k,j}\big)^\top\mathcal{B}_{k,j}, \\
        \mathcal{G}_{k,j} &= \big(\mathcal{Q}^{\text{u}}_{k,j} + \big(B_k\big)^\top\mathcal{P}_{k+1,j}B_k\big)^{-1}, \\
        \mathcal{B}_{k,j} &= \mathcal{Q}^{\text{ux}}_{k,j} + \big(B_k\big)^\top \mathcal{P}_{k+1,j} A_k,
    \end{aligned}
    \end{equation}
    and parallel forward propagations
    \begin{equation} \label{eq:forward_prop}
        \begin{aligned}
             & \Psix_{j+1,j} = \bar{E}_j, \qquad \Psiu_{k,j} = \mathcal{K}_{k,j} \Psix_{k,j}, \\
             & \Psix_{k+1,j} = (A_k + B_k \mathcal{K}_{k, j}) \Psix_{k,j}. \\
        \end{aligned}
    \end{equation}
    Then, since $\bar E_j$ is right-invertible and \eqref{eq:fastsls_qp_var} represents the system level constraint (SLC) assumed in Lemma \ref{lem:right_invertible}, the controller gains $\Psit$ obey, 
    \begin{equation} \label{eq:forward_prop}
    \hspace{-6pt}\begin{aligned}
         &\Psix_{j+1,j} = \bar{E}_j,~~~~~~~~~\Psiu_{k,j}\bar{E}_j^{\ddagger} = \mathcal{K}_{k,j} \Psix_{k,j}\bar{E}_j^{\ddagger}\\
         &\Psix_{k+1,j}\bar{E}_j^{\ddagger} = \left(A_k + B_k \mathcal{K}_{k,j}\right)\Psix_{k,j}\bar{E}_j^{\ddagger}. \\
    \end{aligned}
\end{equation}
    Then, by unraveling the recursion \eqref{eq:forward_prop}, we have for $k > j$,  
    \begin{align}
        \Psit_{k,j}\bar{E}_j^{\ddagger} 
        = \begin{bmatrix} I_{\nx} \\ \mathcal{K}_{k,j} \end{bmatrix} \left(\prod_{m=j+1}^{k-1}\left(A_m + B_m \mathcal{K}_{m,j}\right)\right)\Psix_{j+1,j}\bar{E}_j^{\ddagger}, \label{eq:unravel_psi_t}
    \end{align}
    while the remaining gains ($\Psi_{k,j}$ where $k \leq j$) are zero, where we define the ordered matrix product $\prod_{i=1}^{n} R_i := R_nR_{n-1}\cdots R_{1}$.
    
    To recover the optimal gains for \eqref{eq:fastsls_qp}, we use the right-invertibility property of $\bar E_j$ and substitute $\Psit_{k,j}\bar{E}^{\ddagger}_j = \TPhit_{k,j}$ into \eqref{eq:unravel_psi_t}: 
    \begin{align}
        \TPhit_{k,j} &= \begin{bmatrix} I_{\nx} \\ \mathcal{K}_{k,j} \end{bmatrix} \left(\prod_{m=j+1}^{k-1}\left(A_m + B_m \mathcal{K}_{m,j}\right)\right)\TPhix_{j+1,j}\nonumber\\ 
        &=\begin{bmatrix} I_{\nx} \\ \mathcal{K}_{k,j} \end{bmatrix} \left(\prod_{m=j+1}^{k-1}\left(A_m + B_m \mathcal{K}_{m,j}\right)\right)I_{\nx}.
    \end{align}
    The second equality stems from $\TPhix_{j+1,j}$ satisfying \eqref{eq:fastsls_qp_x_init}.
    Now, observe \eqref{eq:fastsls_qp_reform} has the same Riccati recursions as \eqref{eq:backward_prop}, thus resulting in an optimal gain, similarly structured to \eqref{eq:unravel_psi_t} where $E_j^{\ddagger}$ is replaced with $I_{\nx}$, $\HPhit$: 
    \begin{align}
        \HPhit_{k,j} = \begin{bmatrix} I_{\nx} \\ \mathcal{K}_{k,j} \end{bmatrix} \left(\prod_{m=j+1}^{k-1}\left(A_m + B_m \mathcal{K}_{m,j}\right)\right)\HPhix_{j+1,j}, \label{eq:unravel_phi_t}
    \end{align}
    where $\HPhix_{j+1,j} = I_{n_x}$, as constrained by \eqref{eq:fastsls_qp_cost_reform_init}. Hence $\forall \,k,j\in[N]~\TPhit_{k,j} =\HPhit_{k,j}$, so $\TPhit =\HPhit$.
\end{proof}

\begin{lemma}\label{lem:right_invertible}
Let, $\bar{\bm{E}} := \textrm{blk\_diag}(I_{\nx}, \bar{E}_0,\dots,\bar{E}_{n-1})$, where for all $i\in[N], \bar E_i \in \mathbb{R}^{\nx \times n_{w}^{(i)}}$ is right-invertible, and $n_{w}^{(i)} \geq n_x$. 
Enforcing the uncertain LTV dynamics
\begin{equation}\label{eq:ltv_proof}
    x_{k+1} = A_kx_k + B_ku_k + \bar{E}_kw_k,~w_k \in \mathcal{E}_{n_w^{(k)}},
\end{equation}
is equivalent to enforcing the system level constraint \eqref{eq:slc}, i.e., \eqref{eq:ltv_proof}$\iff$\eqref{eq:slc},
\begin{equation}\label{eq:slc}
    \begin{bmatrix}
        (I_{N\nx}-\bm{Z}\bm{A}) & -\bm{Z}\bm{B}
    \end{bmatrix}
    \begin{bmatrix}
        \Phix \\ \Phiu
    \end{bmatrix} = \bar{\bm{E}}.
\end{equation}
Furthermore, the controller gain mapping for the SLC in \eqref{eq:slc} is,
\begin{equation}\label{eq:controller_gains}
    \mathcal{K} = (\Phiu \bm E^\ddagger)(\Phix \bm{\bar{E}}^\ddagger)^{-1} = (\Phiu \bm{\bar{E}}^{\top})(\Phix \bm {\bar{E}}^{\top})^{-1},
\end{equation}
where $\bm{\bar{E}}^{\ddagger} = \textrm{blk\_diag}(I,\bar{E}_0^{\ddagger},\dots,\bar{E}_{n-1}^{\ddagger})$.
Lastly, given $\mathcal{K}$, the forward propagation of disturbance gains obeys:
\begin{equation} \label{eq:forward_prop_phi}
    \hspace{-6pt}\begin{aligned}
         &\Phix_{j+1,j} = \bar{E}_j,~~~~~~~~~\Phiu_{k,j}\bar{E}_j^{\ddagger} = \mathcal{K}_{k,j} \Phix_{k,j}\bar{E}_j^{\ddagger}\\
         &\Phix_{k+1,j}\bar{E}_j^{\ddagger} = \left(A_k + B_k \mathcal{K}_{k,j}\right)\Phix_{k,j}\bar{E}_j^{\ddagger}. \\
    \end{aligned}
\end{equation}
\end{lemma}
\begin{proof}
    First, let $\delta_k = E_kw_k$ be such that
    \begin{equation}\label{eq:ltv_proof_equiv}
    x_{k+1} = A_kx_k + B_ku_k + \delta_k.
    \end{equation}
   Now, let $\bm x = [x_0^{\top},\dots,x_{n-1}^{\top}]^{\top},~\bm u = [u_0^{\top},\dots,u_{n-1}^{\top}]^{\top},~\text{and}~\bm{\delta}= [\delta_0^{\top},\dots,\delta_{n-1}^{\top}]^{\top}$. Then, defining the gain mappings $\HPhix$ and $\HPhiu$ such that,
   \begin{equation}\label{eq:gain_map_proof}
       \begin{bmatrix}
           \bm x\\ \bm u
       \end{bmatrix} = \begin{bmatrix}
           \HPhix \\ \HPhiu
       \end{bmatrix} \bm\delta,
   \end{equation}
   we have from \cite{DBLP:journals/arc/AndersonDLM19}: 
   \begin{equation}\label{eq:slc_if}
    \begin{bmatrix}
        (I_{N\nx}-\bm{Z}\bm{A}) & -\bm{Z}\bm{B}
    \end{bmatrix}
    \begin{bmatrix}
        \HPhix \\ \HPhiu
    \end{bmatrix} = I_{N\nx}
\end{equation}
   Defining $\bm{w} = [w_0^{\top},\dots,w_{n-1}^{\top}]^{\top}$, we observe:
   \begin{equation}\label{eq:dist_map}
       \bm{\delta} = \bm{\bar{E}}\bm{w}
   \end{equation}
   Then, the gain mapping from $\bm w$ to $\bm x$ and $\bm u$ via $\Phix$ and $\Phiu$ follows from \eqref{eq:gain_map_proof} and \eqref{eq:dist_map}: 
   \begin{equation}\label{eq:gain_map_proof_number2}
       \begin{bmatrix}
           \bm x\\ \bm u
       \end{bmatrix}  = \begin{bmatrix}
           \Phix \\ \Phiu
       \end{bmatrix}\bm{w} = \begin{bmatrix}
           \HPhix \bm{\bar{E}}\\ \HPhiu\bm{\bar{E}}
       \end{bmatrix} \bm{w} = \begin{bmatrix}
           \HPhix \\ \HPhiu
       \end{bmatrix} \bm{\bar{E}}\bm{w} .
   \end{equation}
   This establishes that $\Phix = \hat\Phix\bm{\bar{E}}$ and $\Phiu = \hat\Phiu\bm{\bar{E}}$. Hence, by multiplying \eqref{eq:slc_if} by $\bm{\bar{E}}$, we arrive at \eqref{eq:slc}. 
   Thus we have shown $\text{\eqref{eq:slc}} \implies \text{\eqref{eq:ltv_proof}}$. To show the equivalence between \eqref{eq:ltv_proof} and \eqref{eq:slc}, we can show $\text{\eqref{eq:ltv_proof}} \implies \text{\eqref{eq:slc}}$, by directly recovering the matrices $A_k, B_k, \bar{E}_k$ $\forall k \in [N]$ from \eqref{eq:slc} to produce \eqref{eq:ltv_proof}. 

   To prove \eqref{eq:controller_gains}, we observe that, 
   $K = \HPhiu(\HPhix)^{-1}$, \cite{DBLP:journals/arc/AndersonDLM19}. Then using the substitutions $\Phix\bm{\bar{E}}^{\ddagger} = \HPhix$ and $\Phiu\bm{\bar{E}}^{\ddagger} = \HPhiu$ from \eqref{eq:gain_map_proof_number2} we arrive at the first equality in \eqref{eq:controller_gains}. The second equality arises from algebraic manipulation using the definition $\bm{ \bar{E}}^{\ddagger} = \bm{ \bar{E}}^{\top}(\bm{ \bar{E}}\bm{ \bar{E}}^{\top})^{-1}$. 
   Lastly, manipulating \eqref{eq:controller_gains} to see $\mathcal{K}(\Phix \bm{\bar{E}}^\ddagger) = \Phiu \bm{\bar{E}}^\ddagger$, and expanding the system level constraint in \eqref{eq:slc} we have that
\begin{equation} \label{eq:forward_prop_phi_id}
    \hspace{-6pt}\begin{aligned}
         &\Phix_{j+1,j} = \bar{E}_j,~~~~~~~~~\Phiu_{k,j}\bar{E}_j^{\ddagger} = \mathcal{K}_{k,j} \Phix_{k,j}\bar{E}_j^{\ddagger}\\
         &\Phix_{k+1,j}\bar{E}_j^{\ddagger} = \left(A_k + B_k \mathcal{K}_{k,j}\right)\Phix_{k,j}\bar{E}_j^{\ddagger}. \\
    \end{aligned}
\end{equation}
\end{proof}

\section{Proofs}\label{app:proofs}

\setcounter{theorem}{3}
\begin{theorem}
Let $\mathbf{z}$, $\mathbf{v}$, and $\Phit$ be a feasible solution for \eqref{eq:fastsls_problem}. Then, \eqref{eq:sls_reach_non_center} is guaranteed to overapproximate the true closed-loop reachable set \eqref{eq:reach}, i.e., $\Omega_k \subseteq \bar\Omega_k$ for all $k \in [N]$.
\end{theorem}
\begin{proof}
    We will first prove \eqref{eq:sls_reach_non_center} is correct assuming valid disturbances. 
    Let $\mathcal{W}_k$ specify the disturbance $\forall k \in [N]$. Then, from SLS \cite{leeman2025robust}, we have that the overapproximation of the reachable set is, 
    \begin{equation}
    \begin{aligned}\label{eq:sls_reach_prelim}
        \bar\Omega_k^x := z_k +\textstyle\bigoplus_{j=0}^{k-1}\Phix_{k,j}\mathcal{W}_j,\\\bar\Omega_k^u := v_k + \textstyle\bigoplus_{j=0}^{k-1}\Phiu_{k,j}\mathcal{W}_j.
    \end{aligned}
    \end{equation}
    Now, suppose $\mathcal{W}_k$ can be expressed as a zonotope, $\bar{c}_k \oplus \bar{E}_k\mathcal{E}_{2n_x}$. Making this substitution in \eqref{eq:sls_reach_prelim} we arrive at,
    \begin{equation}
    \begin{aligned}\label{eq:sls_reach_prelim2}
        \bar\Omega_k^x := z_k +\textstyle\bigoplus_{j=0}^{k-1}\Phix_{k,j}(\bar{c}_k \oplus \bar{E}_k\mathcal{E}_{2n_x}),\\\bar\Omega_k^u := v_k + \textstyle\bigoplus_{j=0}^{k-1}\Phiu_{k,j}(\bar{c}_k \oplus \bar{E}_k\mathcal{E}_{2n_x}).
    \end{aligned}
    \end{equation}
    Using the associativity and commutativity of the Minkowski sum \eqref{eq:sls_reach_prelim2} can be rewritten into,

    \begin{equation}
    \begin{aligned}\label{eq:sls_reach_prelim3}
    \bar\Omega_k^x := z_k + \textstyle\bigoplus_{j=0}^{k-1}\Phix_{k,j}\bar{c}_j \oplus \textstyle\bigoplus_{j=0}^{k-1}\Phix_{k,j}\bar{E}_j \mathcal{E}_{2\nx},
    \\\bar\Omega_k^u := v_k+ \textstyle\bigoplus_{j=0}^{k-1}\Phiu_{k,j}\bar{c}_j \oplus \textstyle\bigoplus_{j=0}^{k-1}\Phiu_{k,j}\bar{E}_j \mathcal{E}_{2\nx}.
    \end{aligned}
    \end{equation}
    Then recalling that $\textstyle\bigoplus_{j=0}^{k-1}\Phix_{k,j}\bar{c}_j$ is a Minkowski sum of singleton sets, we can replace it with $\textstyle\sum_{j=0}^{k-1}\Phix_{k,j}\bar{c}_j$ (and similarly do this for $\textstyle\bigoplus_{j=0}^{k-1}\Phiu_{k,j}\bar{c}_j$), and arrive at at the disturbance reachable set \eqref{eq:sls_reach_non_center}. Now, we will show the disturbance bounds, $\mathcal{Z}_r(z_j,v_j, \bm{\mu}_j,\bm{\sigma}_j)$, $\mathcal{Z}_c(z_j,v_j, \bm{\mu}_j,\bm{\sigma}_j)$, in \eqref{eq:rmpc_linearization_radius}-\eqref{eq:rmpc_linearization_center} are valid overapproximations $\mathcal{\bar E}_k$ \eqref{eq:lin_bound} and certify that the closed-loop reachable sets $\bar\Omega^{x}_k$ and $\bar\Omega^{u}_k$ overapproximate the true reachable set $\Omega^x_k$ and $\Omega^u_k$ \eqref{eq:reach}.

    \underline{\textit{Base Case:}} The disturbance at timestep $k=0$ is precisely $E(z_0) = E(\bar x_0)$, which is captured by \eqref{eq:zonotope_combination}, and thus \eqref{eq:rmpc_linearization_radius}-\eqref{eq:rmpc_linearization_center} provides overapproximate disturbance bounds. 

    \underline{\textit{Inductive Step:}} Suppose that for all timesteps $j \in [k]$, $d_j(x_j, u_j, w_j) \in E_j\mathcal{E}_{\nx} \oplus \bar R_j \mathcal{E}_{\nx} \oplus \bar c_j$. Then, we have a overapproximate reachable sets, $\Omega_k^x \subseteq \bar\Omega^{x}_{k}$ and $\Omega_k^u \subseteq \bar\Omega^{u}_{k}$, at timestep $k$. Define $\bar\Omega_k := \{[x^{\top},u^{\top}]^{\top} \mid x \in \bar\Omega^{x}_{k}, u\in\bar\Omega^{u}_{k}\}$. Then observe that $\bar\Omega_k \subseteq \cX_k := \left([z_k^\top, v_k^\top]^\top + \bm{\mu}_k\right) \oplus \textrm{diag}(\bm{\sigma}_k)\mathcal{E}_{\nx+\nuu}$, as defined in \eqref{eq:zontope_combination_interval}. Then, by construction, for all $[x_k^\top,u_k^\top]^\top \in \bar{\Omega}_{k}$ and for all $w_k \in \mathcal{E}_{n_x}$, we have that $r(x_k, u_k, w_k)\in \bar{c}_k \oplus \bar{R}_k\mathcal{E}_{n_x}$, where $(\bar{c}_k,\bar{R}_k)$ is computed using $\cX_k$ in \eqref{eq:zonotope_combination}. Similarly, the nominal disturbance $E_kw_k \in E(z_k)\mathcal{E}_{n_x}$. Thus, the total disturbance \eqref{eq:linearization_error_all},
    $r_k^\textrm{lin} + r_k^{\mathrm{dist}} + E_k w_k = d(x_k,u_k,w_k) \in \bar{c}_k \oplus \bar{R}_k\mathcal{E}_{n_x} \oplus E(z_k)\mathcal{E}_{n_x}$. Concatenating the two generator matrices produces \eqref{eq:zonotope_combination_concat}. Thus, the disturbance bound at timestep $k$ in \eqref{eq:rmpc_linearization_radius}-\eqref{eq:rmpc_linearization_center} is an overapproximation and contains the true residual error, i.e., $d_k(x_k, u_k, w_k)$ \eqref{eq:linearization_error_all}. 

    Having established valid disturbance bounds for all steps $k \in [N]$, SLS guarantees that the overapproximation of the closed-loop reachable set \eqref{eq:sls_reach_non_center}, $\bar\Omega^{x}_{k}$ and $\bar\Omega^{u}_{k}$, are valid, i.e., $\Omega^x_{k} \subset\bar\Omega^{x}_{k}$ and $\Omega^u_{k} \subset\bar\Omega^{u}_{k}$ for all $k \in [N]$.  
\end{proof}

\section{GPUSLS Duals $\tau$}\label{app:sls_consistency}
As discussed in Sec.~\ref{sec:method_nonlinear_sls}, the RNOCP problem is solved by iteratively alternating between solving a nominal trajectory optimization and a controller update. Since these two optimizations are solved separately, we use the dual variable $\tau$ to enforce consistency between the two optimizations as discussed in \cite[App. B]{fang2026safe}. Due to our reformulation, we modify the procedure as follows. As with \cite[App. B]{fang2026safe}, we define the auxiliary terms $\beta_{k,j} \in \mathbb{R}^{n_c}$ and $\beta_{N,j} \in \mathbb{R}^{n_f}$, where
\begin{subequations}
    \begin{align}
        \beta_{k,j} &= \left[ G_k \mathbf{\Phi}_{k,j} \bar{c}_j + \|G_k \mathbf{\Phi}_{k,j}\bar{E}_j\|_{1,\text{r}}\right]^{\circ 2} \\ 
        &\forall j \in [N], \qquad \forall k \in [j, N] \\
        \beta_{N,j} &= \left[ G_N\mathbf{\Phi}^{\text{x}}_{N,j}\bar{c}_j + \| G_N \mathbf{\Phi}^\text{x}_{N, j} \bar{E}_j \|_{1, \text{r}} \right] ^ {\circ 2} \\
        & \forall j \in [N]
    \end{align}
\end{subequations}
where $[\cdot]^{\circ 2}$ denotes element-wise square. We can then calculate our dual variables $\tau$ by  
\begin{equation}
    \tau_{k,j} = \frac{\bar{\lambda}_k}{\sqrt{\beta_{k,j} + \epsilon}} \quad \forall j \in [N], \quad \forall k \in [j, N]
\end{equation}
where we denote $\bar{\lambda}_k$ as the dual variable of the nominal trajectory optimization (corresponding to the dual variable $\lambda_k$ in~\cite{fang2026safe}), introduced here to avoid notation conflict. Finally, we can define our cost terms as
\begin{equation}\label{eq:sls_cost_matrices_tau}
    \begin{aligned}
        \hspace{-5pt}& \mathcal{Q}_{k,j} = \left( \text{diag}\left( \sqrt{\tau_{k,j}}\right)
        G_k,
        \begin{bmatrix}
            \tilde{Q}^{1/2} & 0\\
            0 & \tilde{R}^{1/2}
        \end{bmatrix}
        \right), \\
        \hspace{-5pt}& \mathcal{Q}_{N,j} = \left( \text{diag}\left( \sqrt{\tau_{N,j}}\right) G_N, \tilde{Q}_N^{1/2}\right).
    \end{aligned}
\end{equation}
Following \cite{leeman2024fast} and \cite{fang2026safe}, this formulation preserves dual consistency and yields an equivalent reformulation of the RNOCP under the proposed decomposition.
\end{document}